\documentclass[11pt,letterpaper]{article}

\usepackage[margin=.96in]{geometry}

\usepackage{amsmath,amsfonts,graphpap,amscd,mathrsfs,graphicx,lscape}
\usepackage{epsfig,amssymb,amstext,xspace}
\usepackage{algorithm,algorithmic}

\usepackage{color}              
\usepackage{epsfig}

\usepackage[suppress]{color-edits}
\addauthor{et}{green}
\addauthor{vs}{red}
\addauthor{bl}{blue}
\addauthor{mf}{magenta}
\addauthor{ni}{red}

\usepackage{amsthm}

\newcommand{\SubmissionOmit}[1]{#1} 
\newcommand{\FullversionOmit}[1]{} 



\newcommand{\E}{\mathbb{E}}

\newcommand{\rev}{\mathcal{R}}

\newcommand\Tstrut{\rule{0pt}{5ex}}

\newtheorem{theorem}{Theorem}
\newtheorem{lemma}[theorem]{Lemma}

\newtheorem{defn}[theorem]{Definition}

\def\squareforqed{\hbox{\rule{2.5mm}{2.5mm}}}

\def\QED{\ifmmode\squareforqed 
  \else{\nobreak\hfil   
    \penalty50                 
    \hskip1em                  
    \null                      
    \nobreak                   
    \hfil                      
    \squareforqed              
    \parfillskip=0pt           
    \finalhyphendemerits=0     
    \endgraf}                  
  \fi}

\def\blksquare{\rule{2mm}{2mm}}
\def\qedsymbol{\blksquare}
\newcommand{\bg}[1]{\medskip\noindent{\bf #1}}
\newcommand{\ed}{{\hfill\qedsymbol}\medskip}

\newenvironment{proofof}[1]{{\it{Proof of #1 : }}}{\ed}

\newenvironment{example}{\bg{Example. }}{\ed}

\newcommand{\R}{\ensuremath{\mathbb R}}

\newcommand{\N}{\ensuremath{\mathbb N}}

\newcommand{\F}{\ensuremath{\mathcal F}}

\newcommand{\floor}[1]{\ensuremath{\left\lfloor#1\right\rfloor}}

\newcommand{\comment}[1]{}
 {}

\newcommand{\junk}[1]{}

\newlength{\tmp} \newlength{\lpsx} \newlength{\lpsy} \newlength{\upsx} \newlength{\upsy}

\newcommand{\poa}{\text{\textsc{PoA}} }

\newcommand{\opt}{\text{\textsc{Opt}} }

\newcommand{\specialcell}[2][c]{%
  \begin{tabular}[#1]{@{}c@{}}#2\end{tabular}}

\newcommand{\lm}{\ensuremath{\lambda,\mu}}

\newcommand{\Gn}{\ensuremath{\{G^n\}}}
\newcommand{\Mn}{\ensuremath{\{\M^n\}}}

\newcommand{\Seq}[1]{\ensuremath{\{#1^n\}}}

\newcommand{\SC}{\ensuremath{SC}}

\newcommand{\M}{\ensuremath{\mathcal M}}

\newcommand{\V}{\ensuremath{\mathcal V}}

\newcommand{\BCCE}{\text{\textsc{BNE}}}

\newcommand{\B}{\ensuremath{\mathcal{B}}}

\begin{document}

\title{The Price of Anarchy in Large Games}

\author{Michal Feldman\thanks{Tel-Aviv University, {\tt mfeldman@post.tau.ac.il}}
\and
Nicole Immorlica\thanks{Microsoft Research, {\tt nicimm@microsoft.com}}
\and
Brendan Lucier\thanks{Microsoft Research, {\tt brlucier@microsoft.com}}
\and
Tim Roughgarden\thanks{Stanford University, {\tt tim@cs.stanford.edu}}
\and
Vasilis Syrgkanis\thanks{Microsoft Research, {\tt vasy@microsoft.com}}
}
\date{}

\maketitle

\begin{abstract}
Game-theoretic models relevant for computer science applications usually feature a large number of players. The goal of this paper is to develop an analytical framework for bounding the price of anarchy in such models. We demonstrate the wide applicability of our framework through instantiations for several well-studied models, including simultaneous single-item auctions, routing games, and greedy combinatorial auctions. We identify conditions under which the POA of large games is better than that of worst-case instances. Our results also give new senses in which simple auctions can perform almost as well as optimal ones in realistic settings.
\end{abstract}

\thispagestyle{empty}
\newpage 
\setcounter{page}{1}

\section{Introduction}

Game theory is an indispensable tool for reasoning about
a wide range of computer science applications, such as routing,
network formation, bandwidth pricing, and automated auctions.
However, the game-theoretic models relevant for these applications differ
from the classical examples that grace the opening pages of every
textbook on the subject.
Most of the latter are games with only few, or even two, players ---
the Prisoner's Dilemma, Battle of the Sexes, Matching Pennies, and so
on.
This focus on games with a small number players dates back to the origin
of game theory~\cite{vNM}.
In typical computer science applications, the number of players is
{\em large}.

Are games with many players harder or easier to understand than those
with few?
One answer is ``harder:''
the normal-form representation of a game grows exponentially with the
number of players~$k$, and $k$-player games can be viewed as
special cases of $(k+1)$-player games.
But for many natural classes of games, we might hope that the answer
is ``easier:'' perhaps the influence of each player on the outcome is
small, which in turn makes their optimal behavior easy to
characterize.  For example, in many market settings with a large
number of small players, the action of a single player has only
negligible effect on prices, and thus players can be
accurately modeled as ``price-takers.''


The goal of this paper is to develop an analytical framework for bounding
the inefficiency of equilibria --- the {\em price of anarchy (POA)}
--- in well-motivated classes of games with a large number of players.
In all of the specific models studied in this paper,
worst-case POA is well understood.
These bounds, like any worst-case bounds, tend to be overly
pessimistic and are determined by pathological examples.
Our aim here is to show qualitatively better POA bounds for these
models, assuming only that the games are ``large.''

Our framework and results (detailed below) are chosen with
four
desiderata in mind.
%
First, the framework should be general enough to encompass many
different models in which the POA has been studied.
%
Second, the framework should be relatively easy to use, to maximize
its future applicability.
%
Third, the framework should provide new insights about which mechanisms
are likely to perform well in realistic settings.
%
Finally, the framework should differentiate notions of "largeness"
that lead to smaller inefficiency from those that do not.


\subsection{Summary of Results}

We define {\em $(\lm)$-smooth in the large} game sequences, and show
that the POA of games in such a sequence approaches
$\lambda/(1-\mu)$.  This notion inherits the generality and robustness
of previous smoothness definitions~\cite{Roughgarden2009,Roughgarden2012,Syrgkanis2012,Syrgkanis2013}: many different
applications are amenable to a smoothness-type analysis, and the
resulting POA bounds apply to a wide range of equilibria, such as
correlated and Bayes-Nash equilibria.

We also propose an intuitive sufficient condition for a game sequence
to be smooth in the large.  This condition is relatively easy to
apply, and  most of our specific results are derived from it.
The idea is to formalize the intuition that large
games are more efficient because no individual can significantly
affect the game's outcome.
Precisely, one defines for each player an
{\em approximate utility}, intended to represent the utility a player
expects to receive after incorrectly assuming that
his strategy has no effect on his utility.\footnote{This is purely for
the sake of analysis; the resulting POA bound
applies to the equilibria of the actual
game.}
The sufficient condition requires that the
approximate utility is  $(\lambda,\mu)$-smooth with respect to
the actual game.
We prove that if this condition is satisfied,
then the game sequence is $(\lambda,\mu)$-smooth in the large.

We apply our definitions to obtain POA bounds for ``large versions" of
a number of well-studied models.
\niedit{Our flagship example is simultaneous single-item auctions where we obtain full efficiency in the large even with general combinatorial valuations.  We then demonstrate the versatility of our framework by applying it to routing games.  In the appendix, we also consider greedy combinatorial auctions.}

\textbf{Simultaneous Uniform Price Auctions.}
We highlight the application of
our framework to combinatorial auctions.
First,
we show that in a ``large'' combinatorial auction
setting, running separate simultaneous uniform price auctions for each
type of good in the market leads to a fully efficient allocation in
the limit. Specifically, we consider a setting where a fixed set of
$m$ different goods, each of some supply, are auctioned off to a set
of $n$ bidders. Bidders have combinatorial valuations over sets of
allocated units of different items and we assume that they only want
some fixed number $r$ of units from each individual good. We consider
the following way of growing the market: the number of players
increases, the number of units of supply of each good increases with
the number of players, and each player fails to arrive in the market
with some probability $\delta>0$. Under these conditions we show that
the worst-case expected welfare of Bayes-Nash equilibria of these games
converges \niedit{at a rate of $O(1-1/\sqrt{n})$} to the expected optimal welfare. We make no assumptions about the
distributions of the buyers, other than that they are independent.

Our result should be contrasted with worst-case price of anarchy
bounds for simultaneous second price auctions (a special case of
simultaneous uniform price auctions), where for general combinatorial
valuations \niedit{with complementarities} (even when a player wants at most one unit of each item),  such
an auction cannot achieve better than an $O(\sqrt{m})$-approximation in
the worst-case.  A striking feature of this result is that full
efficiency is obtained even though bidders are forced to report far
fewer parameters (polynomial in $m$) than are present in their
(combinatorial) valuations.  Thus the near-optimal equilibria are not
``truthful outcomes'' in any sense.  This highlights the power of
adopting a smoothness-based framework --- we can prove convergence
to full efficiency without ever characterizing what these near-optimal
equilibria look like.


\vsedit{\textbf{Congestion Games.}
For a second application, we apply our framework to congestion games and show a smooth convergence of the price of anarchy of atomic
congestion games to their non-atomic counterparts, as the number of
players grows large.
Specifically, we
analyze the price of anarchy of sequence of congestion games where
each player controls a smaller and smaller fraction of a fixed traffic
rate. We show through the smoothness approach that the price of
anarchy converges at a multiplicative error of $1+o(1/n)$ to the
non-atomic price of anarchy for the same class of latency functions,
where $n$ is the number of players and assuming that each player
controls a flow of $1/n$.   The POA in nonatomic congestion games is
generally much smaller than in atomic congestion games (see
e.g.~\cite{rg}).}

\vsedit{\textbf{Greedy Combinatorial Auctions.}
A second type of auction we analyze with our framework is a greedy
combinatorial auction setting where each player is interested in a
specific bundle of items of size at most $d$ and each bidder might
want many copies of his bundle, but at most some fixed constant
$r$. Similar to the previous auction setting, there is some fixed set of $m$
different goods and each good is in some supply. We consider the
auction where each player submits a set of $r$ marginal bids for his
interest set. The bids are ordered in decreasing order and bids are
allocated in decreasing order as long as they are satisfiable, i.e., no
item in the interest set has run out of supply. We show that as the
number of players grows and the supply of each good grows in
expectation, but is sufficiently uncertain, then the Bayes-Nash
price of anarchy of the greedy mechanism converges to $d$, which
matches the algorithmic approximation of the algorithm without
incentives.
}

\vspace{0.1in}

\noindent\textbf{Discussion of Results.}
We believe that our results are interesting for several reasons.
First, we prove that in many of the game-theoretic models in which the
POA has been studied, the POA in large games is much smaller than the
worst-case bound.
In some (but not all) cases, the POA approaches 1 as the game grows
large and full efficiency is recovered in the limit.
We suspect that our better POA bounds are more
relevant for many computer science settings, which often feature a
large number of ``small'' players in uncertain environments.

Second, our results have interesting implications for mechanism
design.  Theoretically optimal mechanisms can be relatively complex
--- for example, when bidders demand multiple units of a good, the
welfare-maximizing (i.e., VCG) mechanism charges different prices for
different units of the good.
Some of our results give a new sense in which ``simple'' mechanisms
can be near-optimal.  For example, we prove that the (theoretically
suboptimal) uniform-price mechanism has welfare approaching that of
the VCG mechanism as the market grows large.

Third, our results demonstrate that not all notions of a ``large
game'' are equivalent.
Our analysis framework
identifies features of a setting, including plausible types of environmental uncertainty, that lead to more
efficient equilibria.  We elaborate on this point when we discuss
our techniques, below.









\subsection{Our Techniques}


At first blush, it may sound obvious that letting the number of
players tend to infinity results in a ``limit game'' with more
efficient equilibria.  But this point turns out to be delicate: 
this intuition is true for some but not all natural ways of defining
``large'' games.  We use our framework to guide us to assumptions
that plausibly hold in real-world settings and also lead to improved
POA bounds.

For example, in simultaneous uniform price auctions, it is not enough
to merely assume that the number of players is sufficiently large ---
we are assuming no prior distribution, so one can always add ``dummy
players'' to an arbitrary game without affecting the equilibria (or the
POA).  A similar comment applies \niedit{even in an IID Bayesian setting} if we also assume merely that the number of players and available items tend to infinity.
This is illustrated in the following example, taken from \cite{Swinkels2001}.

\vsedit{
\begin{example}(Inefficiency without randomness) Consider a setting where $k$ units of a good are auctioned off to $n=k$ bidders. Each bidder wants at most two units. The value of each bidder for the first unit is the maximum of two  random samples drawn uniformly from the interval $[2,3]$ and independently for each bidder. The marginal value for the second unit is the minimum of the two samples. The units are sold via a uniform price auction: each bidder submits two marginal bids, the $k$ highest marginal bids win and each bidder pays the highest losing marginal bid for each unit he won. We next show that the following bidding is an equilibrium of the game, for any $k$: each bidder submits his higher marginal truthfully and $0$ as his second marginal. The uniform price is $0$ and each bidder derives utility $v_i^1$. For any bidder to get a second unit he needs to bid at least $2$, which would increase the uniform price to at least $2$ on both units. The increase in payment is $4$, while the increase in value is at most $3$, hence not profitable. Under this equilibrium the expected welfare is $k\cdot 2.67$. The expected optimal welfare is the expected sum of the highest $k$ of $2k$ samples from $U[2,3]$, which is approximately $k\cdot 2.75$ as $k$ grows large.
\end{example}}

To make
progress, we build on an idea introduced previously in the economics
literature~\cite{Swinkels2001}, namely {\em probabilistic demand}.  The
idea here is to introduce demand uncertainty by randomly
removing bidders with a small constant probability.  This corroborates
well with many real-world auction settings (especially in automated
auctions, where the participants might well be chosen by a search
engine's heuristic matching algorithm), where bidders cannot be sure
who else will show up and be chosen to participate in an auction.

For another example, in combinatorial greedy auctions, we prove that
the demand uncertainty that we use above is not enough to improve the POA
as much as one would expect.
Intuitively, the reason is that the presence or absence of a single
bidder can have cascading effects on the demands of other bidders
whose demand sets intersect his.
On the other hand, we prove positive results with supply uncertainty
by randomly varying the number of units of each good.  Again, for many
automated auctions (e.g., where the number of goods might correspond
to a number of queries for a given term or the number of views of
a page), this is a quite palatable assumption.

At a technical level, our framework indicates which assumptions lead
to better POA bounds and which do not.
The primary challenge in applying our framework is to define the
approximate utility of players such
that it both approximates the actual utility and is smooth with
respect to the actual game.
Designing these approximate utilities
requires understanding the approximations players may
reasonably sustain regarding the game they are playing in large
markets (again, this is for the analysis only, not a behavorial
assumption).\footnote{Though as a side benefit,
our results also imply that
playing equilibria
with respect to the approximate utilities will form an approximate
equilibria of the large market that retains all the efficiency
properties we prove.  Conceivably, this could be an accurate model of how
players behave in large and complex games.}
For example, in auctions, in our analysis, we think of players as
ignoring
the impact of their bids on the prices.
Our framework works only when
the relevant parameters of the approximate utility
functions approach the real utility functions in the actual game, and
at this level it is analytically tractable to identify which types of
uncertainties are sufficient and which are not.



Finally, our framework demands that
the approximate utilities are smooth with
respect to the actual game.  This step requires some technical
finesse: the smoothness arguments work by defining deviations for
players that are functions of the types.  Often, the deviations in
auctions ask players to bid high values on their optimal allocation.
However, in the presence of noise, the optimal allocation is a random
variable.  One technical contribution is to circumvent
this difficulty by re-interpreting the noise as a type uncertainty in a
Bayesian game.
This step is also possible (with good smoothness parameters) only for
some types of noise.
In particular, Section~\ref{sec:inefficiency-supply} shows that supply uncertainty is {\it
not} sufficient to obtain fully efficient equilibria in
simlutaneous uniform price auctions.  Intuitively, the reason is that
supply uncertainty makes it hard for bidders to figure out what items to bid
on in the deviation, thus introducing a search friction that
persists as the market grows large.\footnote{More generically, supply
uncertainty only works for settings where the deviation in the
smoothness proof depends only on the valuation of the individual
bidder and is independent of the valuation profile of other bidders.}

\subsection{Related Work}

Most previous work in computer science that concentrates specifically on
games with many players is motivated by complexity concerns.  For example,
the literature on ``compact representations" of games proposes succinct
descriptions, with size polynomial in the number of players,
that are well structured but still rich enough to capture many
interesting applications.
See~\cite{kearns07,leytonbrown03,shoham_book} and the
references therein for many examples.
These references also discuss the well-studied problem of computing
equilibria efficiently in compactly represented multi-player games.

A few recent works in computer science and operations research study
somewhat related notions of ``large games," with different goals than
ours.  Kearns et al.~\cite{KPRU14} study games where no player's
action can affect the payoff of any other player by more than a small
amount, and give an algorithm that computes a correlated equilibrium
while satisfying strong incentive and privacy guarantees.
This result was strengthen for large routing games in~\cite{RR14}.
Pai et al.~\cite{PRU14} discuss the extent to which ``folk theorems" of
repeated games continue to hold in such large games.
Recent work that analyzes ``mean field equilibria"
(e.g.~\cite{IJS11}) effectively assumes that
players are small enough that each player can model the rest as
a population, rather than at an individual level.  Classical work in
game theory on ``nonatomic games" (e.g.~\cite{schmeidler73}
and~\cite{aumann-shapley}) also has this flavor.

There is an old if modest tradition in economics of considering large
markets; see~\cite{RP76} for an early example
and Kalai~\cite{Kalai2004} for work on the robustness of
equilibria in large games. Closest to our work is that of Swinkels~\cite{Swinkels2001} who studies a single-good uniform price auction with decreasing marginal valuations and with demand and supply uncertainties under the same large market assumptions that we make. Our work uses a similar intuition regarding the insensitivity of prices in conditions of noisy demand or supply and combines it with techniques taken from the price of anarchy framework to generalize the setting for which full optimality is achieved. Our simultaneous uniform price auction result generalizes the supply uncertainty result of \cite{Swinkels2001} to allow for heterogeneous goods and our greedy auction result generalizes the demand uncertainty result of \cite{Swinkels2001} by allowing bids on bundles of items, rather than single items (the uniform price auction is a special case of the greedy auction for $d=1$). Our framework also allows us to relax some technical assumptions made in \cite{Swinkels2001} regarding the valuation distributions.

The perspective and goals of these works in the economics literature
differ from ours in several predictable ways.
Their emphasis has been on understanding what equilibria ``look like",
and ideally solving for them explicitly (if not in large finite games,
then at least ``in the limit"); a technically difficult subproblem that
often arises in this approach is to prove that the equilibria of large
finite games approach the equilibria of a ``limit game".
Because we care about equilibria only through their objective function
value, we can bypass the problem of characterizing equilibria and
their limits, and instead argue directly about the approximation
guarantees obtained.
In addition, all previous work in economics on efficiency in large
games considered only the special case of full efficiency in the
limit, as in our result on uniform-price auctions.  No previous work
considered models where inefficiency persists in the limit, as with
several of our other results.  Our smoothness-based framework is
general enough to cover both situations with a common analysis.

Other perspectives of large markets were studied recently by Alvzedo and Budish~\cite{Azevedo2011}, who formalized the notion of ``strategyproof in the large" mechanisms, where truthtelling constitutes an approximate equilibrium as many players arrive in the market and submit bids drawn from the same distribution. Our work differs, in that we don't need to make such symmetric strategy assumptions and allows for mechanisms where efficiency can be achieved in the limit even though truthtelling is not necessarily a limit behavior.

The worst-case POA (without the largeness assumption) is well
understood in all of the models that we study.  For the POA of uniform-price auctions and simultaneous item auctions,
see~\cite{Christodoulou2008, Bhawalkar2011, Hassidim2011, Feldman2013, Markakis2012,Markakis2013} .  For the POA of greedy combinatorial auctions,
see~\cite{Lucier2010}.  For the POA of atomic routing games,
see~\cite{Aland2006}.  For the POA of nonatomic routing games,
see~\cite{Roughgarden2002}.

\section{Preliminaries: Mechanisms, Bayes NE and Price of Anarchy}
\label{sec:prelim}
In this work, we present a framework to study the efficiency of games in the large.  We apply this framework to a variety of games, including cost minimization games, value maximization games, and welfare maximization in mechanisms.
For the sake of brevity, we present our framework for the case of mechanisms for combinatorial auction settings.
Analogous frameworks exist for the other settings, and we defer their presentation to Appendix \ref{app:general-games}.

Consider a market with $n$ bidders and $m$ items. Each player $i\in[n]$ has a valuation function $v_i:2^{[m]}\rightarrow \R_+$, that assigns a value for each possible allocation of items. We will denote the set of possible valuations for player $i$ with $\V_i$ and the set of valuation profiles with $\V=\V_1\times\ldots\times \V_n$.

A mechanism $\M$ constists of a triple $\left(\{S_i\}_{i=1}^{n}, \{x_i\}_{i=1}^n, \{P_i\}_{i=1}^{n}\right)$. $S_i$ is a strategy space for each player (and $S=S_1\times \ldots\times S_n$). $x_i:S\rightarrow 2^{[m]}$ is an allocation function that maps a strategy profile to an allocation of items to player $i$, such that $x(s)=(x_1(s),\ldots,x_n(s))$ is feasible (no two items are allocated to different players). $P_i:S\rightarrow \R_+$ is a payment function.  A bidder's utility for an allocation is his value minus his payment, i.e.,
$u_i(s;v_i)=v_i(x_i(s))-P_i(s).$
We will be interested in analyzing the social welfare of an equilibrium strategy profile $s \in S$, which is the total value of the resulting allocation:
\begin{equation}
\textstyle{SW(s;v) = \sum_{i=1}^{n} v_i(x_i(s))}
\end{equation}
The optimal feasible allocation for valuation profile $v$ will be denoted by $\opt(v)$, i.e. $\opt(v)=\max_{x \text{ is feasible}}\sum_{i=1}^n v_i(x_i)$.
The revenue of a mechanism is the sum of the payments, i.e., $\rev(s)=\sum_{i=1}^{n}P_i(s)$.

We will consider a Bayesian setting in which each player's valuation $v_i$ is drawn independently from some distribution $\F_i$.  A strategy function for agent $i$ is a (possibly randomized) mapping $\mu_i$ from $V_i$ to $S_i$, which we think of as a specification of the strategy to use given a valuation.  A Bayesian Nash equilibrium ($\BCCE$) is a profile of strategy functions such that no single agent can increase her expected utility (over randomization in types and strategies) by unilaterally modifying her strategy.  Formally, the profile of strategy functions $\mu=(\mu_1, \ldots, \mu_n)$ is a $\BCCE$ if for all $i$, all valuations $v_i\in \V_i$, and all alternative strategies $s_i' \in S_i$, we have
\[ \E_{v_{-i} \sim \F_{-i}}[u_i(\mu_i(v_i), \mu_{-i}(v_{-i});v_i)] \geq \E_{v_{-i} \sim \F_{-i}}[u_i(s_i', \mu_{-i}(v_{-i});v_i)]. \]
Note that the non-Bayesian notion of Nash Equilibrium is a special case of the above, in which every distribution $\F_i$ is a point mass.

The Bayes-Nash Price of Anarchy ($\BCCE\mbox{-}\poa$) of a mechanism $\M$ is the worst-case ratio between the expected optimal welfare and the expected welfare at equilibrium, over all type distributions and all $\BCCE$. That is,
\begin{equation}
\textstyle{ \BCCE\mbox{-}\poa = \max_{\F} \max_{\mu}\frac{\E_{v \sim \F}[\opt(v)]}{\E_{v \sim \F}[SW(\mu(v);v)]}, }
\end{equation}
where the maximum over strategy functions $\mu$ is taken over all $\BCCE$ for distribution profile $\F$. 

\section{Smoothness in the Large}
\label{sec:framework}
\paragraph{Sequence of mechanisms.}
We will typically work with a sequence of mechanisms $\Mn_{n=1}^{\infty}$, indexed by the number of participating players $n$. For shorter notation we will write $\Seq{x}$ to denote the sequence $\Seq{x}_{n=1}^{\infty}$. 

In a sequence of mechanisms $\Mn$, everything will be changing parametrically with the number of players, such as the set of items $m^n$, the strategy spaces $S^n=(S_1^n,\ldots,S_n^n)$, the allocation functions $x^n=(x_1^n,\ldots,x_n^n)$, the payment functions $P^n=(P_1^n,\ldots,P_n^n)$ and the valuation profile space $\V^n=(\V_1^n,\ldots,\V_n^n)$. We will also denote with $\opt^n(\cdot)$ the optimal welfare, with  $\rev^n(\dot)$ the revenue and with $u_i^n$ the utility of player $i$ in mechanism $\M^n$. For the moment one can imagine arbitrary ways for the mechanism to grow; in subsequent sections we give specific conditions for how the market should grow for our framework to be applicable.

\paragraph{Smoothness in the large.}
For finite games, Roughgarden \cite{Roughgarden2009} introduced the notion of {\it smoothness} as a method for bounding inefficiency of equilibria.  The smoothness approach proceeds by exploring specific deviations, instead of characterizing the (potentially complex) structure of equilibria. This approach was specialized to the mechanism design setting via the notion of smooth mechanisms by Syrgkanis and Tardos \cite{Syrgkanis2013}. We extend the notion of smoothness to large games. In what follows we present the specific extension in the context of mechanism design (i.e., large mechanisms), but the framework is widely applicable and the reader is directed to Appendix \ref{app:general-games} for the formulation of the framework in general games. Intuitively, a sequence of mechanisms is said to be $(\lm)$-\emph{smooth in the large} if for any $\epsilon$, and a sufficiently large number of players, each player $i$ has a special strategy that allows him to acquire a $\lambda\cdot (1-\epsilon)$ fraction of his valuation for his optimal set of items, by paying no more than $\mu$ times the current price paid for these items.

\begin{defn}[Smooth in the large] A sequence of mechanisms $\Mn$ is $(\lambda,\mu)$-smooth in the large if for any $\epsilon>0$, there exists $n(\epsilon)<\infty$, such that for any $n>n(\epsilon)$, for any $v^n\in \V^n$, for each $i\in [n]$, there exists a strategy $s_i^{*,n}\in S_i^n$, such that for any $s^n\in S^n$:
\begin{equation}
\textstyle{\sum_{i=1}^{n} u_i^n(s_{i}^{*,n},s_{-i}^n;v_i^n) \geq \lambda (1-\epsilon) \opt^n(v^n)-\mu\cdot \rev^n(s^n)}
\end{equation}
\end{defn}

The following theorem shows that if a sequence of mechanisms is $(\lambda,\mu)$-smooth in the large, for some $\lambda,\mu\geq 0$, then its price of anarchy as $n\rightarrow \infty$ is at most $\frac{\max\{1,\mu\}}{\lambda}$. Moreover, it implies that for any sufficiently large but finite market of size $n$ the price of anarchy of all Bayes-Nash equilibria is at most a $1+\epsilon(n)$ multiplicative factor away from the limit price of anarchy, where the rate of convergence of $\epsilon(n)$ to $0$ will depend on the application and can be derived from the proof of smoothness in the large.

\begin{theorem} If a sequence of games is $(\lambda,\mu)$-smooth in the large then
\begin{equation*}
\textstyle{\lim\sup_{n\rightarrow \infty} \BCCE\text{-}\poa^n\leq \frac{\max\{1,\mu\}}{\lambda}.}
\end{equation*}
I.e., for any $\epsilon$ there exists a market size $n(\epsilon)$ such that for any $n\geq n(\epsilon)$, every Bayes-Nash equilibrium of the mechanism $\M^n$ with value distributions $\F_1\times \ldots\times \F_n$ has expected social welfare at least $(1-\epsilon)\frac{\lambda}{\max\{1,\mu\}}$ of the expected optimal welfare.
\end{theorem}
\begin{proof}
By $(\lambda,\mu)$-smoothness in the large, for any $\epsilon$ there exists a market size $n(\epsilon)$ such that for any $n\geq n(\epsilon)$
the mechanism $\M^n$ is a $\left(\lambda(1-\epsilon),\mu\right)$-smooth mechanism, as defined in \cite{Syrgkanis2013}. Therefore, by the results in \cite{Syrgkanis2013}, the $\BCCE$-$\poa^n$ is at most $\frac{\max\{1,\mu\}}{\lambda (1-\epsilon)}$. The theorem then follows.
\end{proof}

\subsection{Main Technique: Smooth Approximate Utility Functions}
We present the notion of a $(\lambda,\mu)$-smooth approximate utility function sequence with respect to a sequence of mechanisms $\Mn$.

\begin{defn}[Smooth approximate utility] Let $U_i^n: S^n\times V_i\rightarrow \R_+$ be a utility function for player $i \in [n]$, and let $U^n=(U_1^n,\ldots,U_n^n)$ be a vector of utility functions. A sequence $\Seq{U}$ is a \emph{sequence of $(\lambda,\mu)$-smooth approximate utility functions for the sequence of mechanisms $\Seq{\M}$}  if the following two properties are satisfied:
\begin{enumerate}
\item {\bf(Approximation)}\label{prop:approximation} The approximate utility $U_i^n$ converges to the true utility $u_i^n$ uniformly over $s^n\in S^n$ and $v_i\in \V_i$. I.e., for any $\epsilon$, there exists $n(\epsilon)<\infty$, such that for any $n>n(\epsilon)$, for any $i\in [n]$ and $v_i\in \V_i$, and for any $s^n\in S^n$:
\begin{equation}
\textstyle{\left\|u_i^n(s^n; v_i) -U_i^n(s^n; v_i)\right\|<\epsilon.}
\end{equation}
\item {\bf(Smoothness)} For each mechanism $\M^n$ in the sequence, the approximate utility satisfies the following $(\lambda,\mu)$-smoothness property with respect to $\M^n$:
For any $n$, for any $v\in \V^n$, for any $i\in [n]$, there exists a strategy $s_i^{*,n}\in S_i^{n}$, such that for any strategy profile $s^n\in S^n$:
\begin{equation}
\textstyle{\sum_{i=1}^{n} U_i^n(s_i^{*,n},s_{-i}^n;v_i) \geq \lambda \opt^n(v) -\mu\cdot \rev^n(s^n)}
\end{equation}
\end{enumerate}
\end{defn}

We show that if a sequence of mechanisms admits a $(\lambda,\mu)$-smooth approximate utility sequence, and if its optimal social welfare increases at least at the same asymptotic rate as the number of players, then this sequence of mechanisms is $(\lambda,\mu)$-smooth in the large. The proof is rather straightforward given the definitions, and is deferred to Appendix \ref{app:omitted}\FullversionOmit{ of the full version}.
\begin{theorem}\label{thm:approx-to-large}
If a sequence of mechanisms $\Mn$ admits $(\lambda,\mu)$-smooth approximate utility functions, and $\opt^n(t)=\Omega(n)$, then $\Mn$ is $(\lambda,\mu)$-smooth in the large.
\end{theorem}

\section{Simultaneous Uniform Price Auctions}\label{sec:simultaneous-uniform}
%

We consider a setting with a growing number of $n$ bidders and a fixed number of $m$ different (types of) goods.
There are $k_j^n$ units of each good $j\in[m]$ which grows as $\Omega(n)$ with the number of players.\footnote{In the full version, we show an alternative proof of the theorem that establishes the same result even without the assumption that $k_j^n=\Omega(n)$ (rather assumes that $k_j^n \rightarrow \infty$). However, that proof shows directly that the mechanism is $(1,1)$-smooth in the large, rather than going through the existence of $(1,1)$-smooth approximate utilities.}  Each player $i\in[n]$ has a valuation function $v_i:\N^{m}\rightarrow [0,H]$, that assigns a value for each possible bundle, depending on the number of units of each good.
%
These functions are bounded in their demand for the number of units of each good. 
Specifically, let $x_{i}^{j}$ denote the number of units of good $j$ allocated to player $i$, and let $x_i=(x_{i}^{1},\ldots,x_{i}^{m})$ be an allocation vector for player $i$.  Then there is a publically known constant $r$ such that: $v_i(x_i) = v_i(\min\{x_i^1,r\},\ldots,\min\{x_i^m,r\})$.
We will also assume that these valuations are bounded away from zero for any non-empty allocation, i.e.\ $v_i(x_i)\geq \rho>0$ for every non-zero $x_i$.
%



The units of each good $j\in [m]$ are simultaneously and independently sold via the means of a uniform price auction. The auctioneer solicits $r$ bids $b_{i}^{j,1}\geq \ldots\geq b_{i}^{j,r}$ from each bidder $i$ for each good $j$, referred as marginal bids. All bids of good $j$ (from all players) are ordered in a decreasing order, and each of the first $k_j^n$ bids wins a unit. In the case of ties, bidders are processed in a random order, and all tying bids of a bidder are allocated sequentially in order until the supply of the good runs out. Every player is charged the highest losing marginal bid for good $j$ for every unit of good $j$ allocated to him. We will assume that no bid exceeds some fixed number $B$; i.e., $b_{i}^{j,x}<B$ for every $i,j,x$. Since we assumed that $v_i(x_i)\leq H$, it is a weakly dominated strategy for a player to bid more than $H$ on an individual marginal bid, though our formulation allows even for $B>H$, as long as $B$ doesn't grow with the market \mfcomment{This last sentence is unclear}. We will denote by $\M^n$ an instance of the simultaneous uniform price auction among $n$ players.

Notably, the above auction is not truthful for many reasons. 
First, the auction format is not even rich enough to allow players to express their true valuations, as they are forced to place additively separable bids on the different goods. 
second, even for a single type of good, a uniform-price auction is not truthful for players with multi-unit demands.
Nonetheless, we will show that in large markets, under a particular type of demand uncertainty --- where each bidder ``fails to arrive'' with constant probability --- all equilibria achieve full efficiency.

\begin{theorem}[Full Efficiency in the Limit]\label{thm:simultaneous-uniform}
In the setting described above, if each player fails to arrive in the market with probability $\delta$, then the implied sequence of mechanisms is $(1,1)$-smooth in the large; hence full efficiency is achieved in the limit. \vsedit{Moreover, the fraction of the optimal welfare achieved at equilibrium converges to $1$ at a rate of $1-O\left(\frac{1}{\sqrt{n}}\right)$.}
\end{theorem}

Crucially, the fact that we recover full efficiency in the large is not trivial in our setting.  The result is sensitive, for example, to the type of noise in the system, and different noise may be 
required in different settings.  The noisy arrival of players can be seen as a type of demand uncertainty.  In prior work of Swinkels~\cite{Swinkels2001}, and in Section~\ref{sec:greedy} 
which generalizes the prior work of Swinkels to combinatorial auctions with fixed demand sets, the uncertainty instead regards the supply: namely the probability that the number of units of 
the good equals any fixed number goes to $0$ as the market grows large.  For these settings, supply uncertanty is sufficient to recover full efficiency in the limit.  In contrast, for
simultaneous uniform price auctions, supply uncertainty can lead to a constant factor inefficiency even in the limit. In particular, it sustains a ``search friction'' in the limit: players do not know 
which items will have higher supply and thereby cannot decide which items to target. At equilibrium, their supply prediction ends up leading to constant factor inefficiencies that do not vanish.
A concrete counter example showing that supply uncertainty may lead to constant inefficiency in simultaneous uniform price auctions appears as Example \ref{ex:inefficiency-supply} in Appendix \ref{sec:uniform-app}\FullversionOmit{ of the full version}.

\textbf{Sketch of proof of Theorem \ref{thm:simultaneous-uniform}}
At the high level, Theorem \ref{thm:simultaneous-uniform} is established by showing that the simultaneous uniform price auction where each player fails to arrive with probability $\delta$ admits $(1,1)$-smooth approximate utility functions.
The full proof is deferred to Appendix~\ref{sec:uniform-app}\FullversionOmit{ of the full version}.

The approximate utility functions we define will have the following intuitive interpretation: each player $i$ looks at the $(k+1)$-th highest bid at each auction excluding his own bids. Denote this with $P_{j}^{-i}$. This is the price that the other players would have paid for each unit of good $j$ had player $i$ not been in the market. In player $i$'s approximate utility, he has the delusion that this is also the price he faces; i.e., any marginal bid that he submits that surpasses the price $P_{j}^{-i}$ will win a unit at price $P_j^{-i}$.\footnote{A bid that is equal to $P_j^{-i}$ will pass through the tie-breaking rule.} In the actual market this is obviously not true: to win $x\in\{1,\ldots,r\}$ units, player $i$ actually needs to exceed the $x$-th lowest winning bid in his absence, and his price will be equal to this bid which may be greater than his imagined price of $P_j^{-i}$. However, as we shall soon show, with the proposed noise in the system, the price $P_j^{-i}$ is ''sufficiently random" that it is distributed almost identically to this $x$-th lowest winning bid for any constant $x$.

In what follows we present some of the technical challenges and techniques in our proof.  Following the framework of smooth approximate utilities, we first sketch the proof of the approximation and then the smoothness of the approximate utility functions described above. 

\textbf{Approximation.}
We first show (in Lemma \ref{lem:simultaneous-approximation}) that the player's utility from any bid vector $b$ converges to his approximate utility, as the market grows large. Technically, the two utilities differ either when the allocation is different or when the price paid is different. The allocation differs only when some of the player's marginal bids are among the $k+1-r$ and $k+1$ highest bids,\footnote{In the actual proof we also take care of tie-breaking.} since this is the only case where the player may believe that his marginal bid is a winning bid (under his delusional utility) while it is actually a losing bid (under his true utility). However, due to the random arrival, for any bid $b$, the probability that the number of bids above $b$ is equal to some number $x$ goes to $0$ as $x\rightarrow \infty$. Thus, the probability of any of these events goes to $0$. Now, since there is only a constant number of these events (by the assumption that $r$ is constant), the probability of any bad event occurring goes to $0$ (by the union bound). Finally, we show that the difference in price paid also goes to zero. Technically, the distributions of the $(k+1)$-th and the $(k+1+x)$-th highest bids are identical for any $x\in [-r,r]$. Thus, their expectation converges to zero as well (as they are bounded random variables).

\textbf{Smoothness.}
The other part of the proof (Lemmas \ref{lem:simultaneous-smoothness}, \ref{lem:simultaneous-smoothness-noise}) shows that these delusional utilities satisfy the $(1,1)$-smoothness property. Observe, that under this delusion a player believes that he can always grab his optimal set of items at the current price in which they are sold. This is essentially the $(1,1)$-smoothness property. 
However, there are two crucial subtleties that need to be handled carefully. 
First, the prices of the goods are random, thus unknown to the player. 
Second, the optimal set of items for a player is also random, as it depends on who arrives in the market (which is not observed by the player when he decides his bid vector). 
The first problem is bypassed by observing that since these are threshold price mechanisms, the player can simply bid sufficiently high (even overbid). Specifically, if a player's optimal allocation is $x_i=(x_i^1,\ldots,x_i^m)$, where $x_i^j$ denotes the number of units of good $j$, then by bidding sufficiently high on the $x_i^j$ highest marginal bids on each good $j$, he will almost surely win the items, or otherwise some price must be so high that we can charge the welfare loss to some other allocated player. We will show that bidding $v_i(x_i)$ as the first $x_i^j$ marginal bids on each good $j$ is sufficiently high to establish our $(1,1)$-smoothness argument.
To bypass the second problem, we observe that the utility of any player under this game is lower bounded by the utility if he can bid even when he doesn't arrive, but has value of $0$. The latter game is a simultaneous uniform price auction with no noisy demand, but with a Bayesian uncertainty on the values \mfcomment{this last sentence is unclear}. Thus we will use a technique similar to the one used to show that smoothness for complete information games implies smoothness for games with Bayesian uncertainty in the values \cite{Roughgarden2012,Syrgkanis2012,Syrgkanis2013}. In particular, the smoothness deviation samples an arrival vector from the distribution and uses this random sample as a proxy for the true arrival vector, targeting the optimal bundle under this random sample.

\section{Congestion Games}



%
%
%

\vsedit{We will apply the notion of smoothness in the large for general games (which we formally present in Appendix \ref{app:general-games}) to network routing games. An (unsplittable flow) routing game is specified by a fixed network $H$ and a set of $n$ players.  Each player $i$ has a type corresponding to a source node $u_i$, a destination node $v_i$, and an amount of flow $f_i$.  
A strategy $s_i$ of player $i$ is a path from $u_i$ to $v_i$ in $H$.  Each edge $e$ of $H$ has a corresponding cost function $c_e : [0,1] \to \mathbb R$, where $c_e(x)$ is interpreted as the cost (or delay) of traveling on edge $e$ given that a total flow of $x$ is being routed on that edge.  Typically, each cost $c_e$ is a continuous non-decreasing function.  For a fixed network $H$, let $G^n$ denote an unsplittable flow routing game with $n$ players, in which $f_i = 1/n$ for each player $i$. We can then think of $\Gn$ as a sequence of games, growing large in the sense that individual players control a vanishingly small fraction of the total flow.}


For game $G^n$ (i.e., a particular choice of $n$), the cost function for each player is:
$
c_i^n(s; t_i) = \sum_{e\in s_i} c_e\left(\frac{n_e(s)}{n}\right),
$
where $n_e(s)=|\{j: e\in s_j\}|$ is the number of players that select, as their strategy, a path that includes edge $e$.
%
For this game, we consider the following \emph{approximate} cost function:
\begin{equation}\label{eqn:approx-congestion}
C_i^n(s; t_i) = \sum_{e\in s_i} c_e\left(\frac{n_e(s_{-i})}{n}\right),
\end{equation}
where $n_e(s_{-i})=|\{j\neq i: e\in s_j\}|$. Essentially, this approximation corresponds to a player ignoring his own effect on the delay of a link when calculating his total delay for the links that he uses. Though this assumption might lead to large errors in a small game, it is a good approximation in a large game context where each individual player controls only a vanishing fraction of the traffic.

\begin{lemma}\label{lem:approx-cost}
If cost functions are continuous and satisfy:
$
x\cdot c_e(y)\leq \lambda\cdot x\cdot c_e(x)+\mu\cdot y\cdot c_e(y)
$
then the sequence of functions $\{C_i^n\}$ defined in Equation \eqref{eqn:approx-congestion} is a valid sequence of $(\lambda,\mu)$-smooth approximate cost functions for the sequence of congestion games $\{G^n\}$. \SubmissionOmit{(proof in Appendix \ref{app:congestion})}
\end{lemma}
For example, when the edge cost functions are affine, then the property in Lemma \ref{lem:approx-cost} 
is known to hold for $\lambda=1$ and $\mu=\frac{1}{4}$, which allows us to conclude that congestion games with linear cost functions are $\left(1,\frac{1}{4}\right)$-smooth in the large.
Thus, as long as $\opt^n = \Omega(n)$ (which occurs as long as the linear coefficients of the cost functions are bounded away from $0$), we conclude that the price of anarchy 
converges to the well-known $4/3$ bound for non-atomic routing games \cite{Roughgarden2002}. More generally, polynomial latency functions satisfy the property for $\lambda=1$ and $\mu=d(d+1)^{-(d+1)/d}$ which leads to a price of anarchy in the limit of  $\Theta(d/\log(d))$. 
This is in stark comparison with the price of anarchy for atomic congestion games with degree $d$ polynomial delay functions, which is $\Theta\left(\left(d/\log(d)\right)^d\right)$.

Finally, if we also assume that the cost functions are $L$-lipschitz for some constant $L$, then the proof of Lemma \ref{lem:approx-cost} implies that the convergence rate to the non-atomic price of anarchy is of the order of $1/n$, i.e. if the cost functions satisfy the property in Lemma \ref{lem:approx-cost} then for the sequence of games $BNE-PoA^n = \frac{\lambda}{1-\mu}+O(1/n)$. 


\bibliographystyle{plain}
\bibliography{large-bib}

\newpage

\begin{appendix}

\section{Omitted Proofs}\label{app:omitted}

\begin{proofof}{Theorem \ref{thm:approx-to-large}}
Since $\Mn$ admits $(\lambda,\mu)$-smooth approximate utility functions $\Seq{U}$, we have that for any $n$ and $v\in \V^n$ there exists strategies $s_i^{*,n}$ for each $i\in [n]$ such that, for any $s\in S^n$,
\begin{equation*}
\sum_{i=1}^{n} U_i^n(s_i^{*,n},s_{-i};v_i) \geq \lambda \opt^n(v) -\mu\cdot \rev^n(s).
\end{equation*}
By the approximation property of $U_i^n$ we have that for any $\epsilon$, there exists $n(\epsilon)<\infty$ such that for any $n>n(\epsilon)$: $u_i^n(s;v_i)\geq U_i^n(s;v_i)-\epsilon$ for any $v_i\in \V_i^n$ and $s^n\in S^n$. Thus:
\begin{equation*}
\sum_{i=1}^{n} u_i^n(s_i^{*,n},s_{-i};v_i) \geq \lambda \opt^n(v) -\mu\cdot \rev^n(s) - n\cdot \epsilon.
\end{equation*}
Since $\opt^n(v)=\Omega(n)$, for any $v\in \V^n$, we can write $\opt^n(v)\geq \rho \cdot n$ for some $\rho>0$ and for sufficiently large $n$. Thus we get:
\begin{equation*}
\sum_{i=1}^{n} u_i^n(s_i^{*,n},s_{-i};v_i) \geq \left(\lambda-\frac{\epsilon}{\rho}\right) \opt^n(t^n) -\mu\cdot \rev^n(s).
\end{equation*}
Therefore, for any $\delta>0$, we can pick $\epsilon$ appropriately small, such that $\lambda-\frac{\epsilon}{\rho}\geq \lambda(1-\delta)$, which would then yield the theorem.
\end{proofof}

\vsdelete{\begin{figure}
\begin{center}
\begin{tabular}{|c | c| c| c|}
        \hline\noalign{\smallskip}
Model  & Worst-Case PoA & Large Market PoA & Section\\
\hline\hline
\specialcell{Simultaneous Uniform\\ Price Auctions} & \specialcell{$\Omega(\sqrt{m})$ (General Valuations) \cite{Roughgarden2014}\Tstrut \\
$\leq m$}  & \specialcell{=1 (General Valuations)\\ (full efficiency in the limit)} & Section \ref{sec:simultaneous-uniform} \\ 
\hline
\specialcell{Greedy Combinatorial\\ Auctions} & \specialcell{$\leq d+1$\\ (Single-Minded bidders,\\ size-$d$ bundles)\\ $>d$} \cite{Lucier2010} & \specialcell{$\leq d$ \\ (matching algorithmic\\ approximation ratio)}  & Section \ref{sec:greedy}\\ 
\hline
Congestion Games & \specialcell{$\Theta\left(\left(\frac{d}{\log(d)}\right)^d\right)$\cite{Aland2006} \\ Degree $d$ polynomial latencies} & \specialcell{$\Theta\left(\frac{d}{\log(d)}\right)$\\ (matching non-atomic\\ routing bound \cite{Roughgarden2002})}  & Section \ref{sec:congestion-games}\\
\hline
\end{tabular}
\caption{Comparison of large market price of anarchy results with worst-case price of anarchy.}\label{table:summary} 
\end{center}
\end{figure}
}

\section{Simultaneous Uniform Price Auctions}
\label{sec:uniform-app}

\paragraph{Proof of Theorem \ref{thm:simultaneous-uniform}.} We will view the simultaneous uniform price auction with random arrivals as an ex-ante mechanism $\M^{n,\delta}$, where the noise is endogenized in the rules of the mechanism and then we will show that mechanism $\M^{n,\delta}$ is $(1,1)$-smooth in the limit. We will refer to this mechanism as \emph{simultaneous uniform price auction with endogenous $\delta$-noisy demand}.

\paragraph{Basic Notation.} We first introduce some useful notation. We will denote with $u_i^n(b;v_i)$ the expected utility from a simultaneous uniform price auction where $b=(b_1,\ldots,b_n)$ and $b_i$ is a vector of marginal bids $b_i^{j,x}$, with $j\in[m]$ and $x\in [r]$, satisfying the decreasing marginal bid property, i.e. $b_i^{j,x}$ is decreasing in $x$.
We will denote with $x_i(b)$ the allocation of player $i$ under bid profile $b$ in the simultaneous uniform price auction, which is a random variable (due to tie-breaking). For any vector $x$, we will denote with $\theta_{t}(x)$, the $t$-th highest element in $x$. Thus $\theta_{t}(b^j)$ is the $t$-th highest marginal bid at the uniform price auction for good $j$. Thus we can write:
\begin{equation}
u_i^n(b;v_i) = \E\left[v_i(x_i(b))-\sum_{j\in [m]} x_i^j(b)\cdot \theta_{k_j^n+1}(b^j)\right]
\end{equation}
where expectation is taken over $x_i(b)$.

We will denote with $u_i^{n,\delta}(b;v_i)$ the expected utility of player $i$ in the simultaneous uniform price auction with noisy arrivals. Concretely, let $z_i$ be a $\{0,1\}$ random variable that equals $1$ with probability $1-\delta$, indicating whether player $i$ arrived in the market and let $z=(z_1,\ldots,z_n)$. Then
\begin{equation}
u_i^{n,\delta}(b;v_i) = \E[z_i\cdot u_i^{n}(b\cdot z;v_i)].
\end{equation}

\paragraph{Approximate Utility.} We denote with $U_i^n(b;v_i)$ an approximate utility associated with the non-noisy sequence of mechanisms, defined by the following allocation and payment rules (due to heavy notation we avoid giving an algebraic description of $U_i^n$ and only describe it in words). We remind the reader the approximate utility is not the utility associated with any mechanism, and in fact would not be feasible for all bidders simultaneously.  It is simply a construct for the proof of smoothness, and can be interpreted as an intuition for what's guiding bidder behavior. To construct the approximate utility , for each uniform price auction $j\in [m]$, let $\theta_{k_j^n+1}(b_{-i}^j)$ be the $k_j^n+1$-highest marginal bid excluding player $i$'s bids. Every marginal bid $b_i^{j,x}>\theta_{k_j^n+1}(b_{-i}^j)$ wins a unit at auction $j$ and bids with $b_{i}^{j,x}=\theta_{k_j^n+1}(b_{-i}^j)$ win with some probability that follows from the random tie-breaking rule described in the beginning of the section. We will denote with $X_i(b)$ the allocation function that is implied by the above description, which is also a random variable due to the tie-breaking rule. For every unit that a player $i$ wins at auction $j$, she pays $\theta_{k_j^n+1}(b_{-i}^j)$. Thus we can write the approximate utility as:
\begin{equation}
U_i^n(b;v_i) = \E\left[v_i(X_i(b))-\sum_{j\in [m]} X_i^j(b)\cdot \theta_{k_j^n+1}(b_{-i}^j)\right]
\end{equation}

Then denote with $U_i^{n,\delta}(b;v_i)$ an approximate utility for the noisy arrival mechanism, which is simply defined as:
\begin{equation}
U_i^{n,\delta}(b;v_i) = \E\left[z_i\cdot U_i^{n}(b\cdot z;v_i)\right]
\end{equation}

\paragraph{$(1,1)$-Smoothness of Approximate Utility.} We will first show that the approximate utility $U_i^{n,\delta}$ satisfies the $(1,1)$-smoothness property with respect to the sequence of mechanisms $\M^{n,\delta}$. To achieve this we will break it into two parts. First we will show that the approximate utility $U_i^n$, satisfies the $(1,1)$-smoothness property with respect to the non-noisy sequence of mechanism $\M^n$. Then we show generically, that if a sequence of utility functions $U_i^n(s;v_i)$, satisfy the $(\lambda,\mu)$-smoothness property with respect to a sequence of mechanisms $\M^n$, then the sequence of utility functions $U_i^{n,\delta}(s;v_i)=\E\left[z_i\cdot U_i^n(s\cdot z;v_i)\right]$ satisfies the $(\lambda,\mu)$-smoothness property with respect to the sequence of mechanisms $\M^{n,\delta}$, which is the version of $\M^n$ where each player arrives with probability $\delta$. This completes the first part of the proof.

\begin{lemma}\label{lem:simultaneous-smoothness}
$U_i^n$ satisfies the $(1,1)$-smoothness property with respect to the sequence of simultaneous uniform price auctions $\M^n$.
\end{lemma}
\begin{proof}
Consider a mechanism $\M^n$ in the sequence, valuation profile $v\in \V^n$ and let $\opt^n(v)$ be the optimal allocation. For each player $i$ let $x_i^*$ denote his allocation in the welfare maximizing allocation. Consider the following deviation $b_i^*$ for each player $i$: at each auction $j\in [m]$, bid $v_i(x_i^*)$ as the first $x_i^{j,*}$ marginal bids and $0$ on the remaining marginal bids.

Consider any bid profile $b$. There are two cases: either player $i$ wins at least his optimal allocation in his dellusion in which case he gets approximate utility:
\begin{align*}
U_i^n(b_i^*,b_{-i}) \geq v_i(x_i^*) - \sum_{j\in [m]} x_i^{j,*}\cdot \theta_{k_j^n+1}(b_{-i}^j)
\end{align*}
or otherwise, there is at least one $q\in [m]$ with $x_i^{q,*}>0$, for which $\theta_{k_q^n+1}(b_{-i}^q)\geq v_i(x_i^*)$ and at which player $i$ wins strictly less than $x_i^{q,*}$ units. In that case, player $i$'s approximate utility is at least:
\begin{align*}
U_i^n(b_i^*,b_{-i};v_i) \geq -\sum_{j\in [m]} x_i^{j,*} \cdot \theta_{k_j^n+1}(b_{-i}^j)+ \theta_{k_q^n+1}(b_{-i}^q) \geq v_i(x_i^*) - \sum_{j\in [m]} x_i^{j,*}\cdot \theta_{k_j^n+1}(b_{-i}^j)
\end{align*}
Hence, the latter inequality holds always. Summing up the inequality for each player and observing that $\theta_{k_j^n+1}(b_{-i}^j)\leq \theta_{k_j^n+1}(b^j)$  and $\sum_{i=1}^{n} x_i^{j,*}\leq k_j^n$, we get:
\begin{align*}
\sum_{i=1}^{n} U_i^n(b_i^*,b_{-i};v_i)  \geq \opt^n(v) -\sum_{j\in [m]} k_j^n\cdot \theta_{k_j^n+1}(b^j)
\end{align*}
Now it is easy to see that $\rev^n(b) = \sum_{j\in [m]} k_j^n\cdot \theta_{k_j^n+1}(b^j)$, since at each uniform price auction, either $k_j^n$ units were sold at a price of $\theta_{k_j^n+1}(b^j)$ or $\theta_{k_j^n+1}(b^j)=0$. This completes the proof.
\end{proof}

\begin{lemma}\label{lem:simultaneous-smoothness-noise}
If $U_i^n$ satsifies the $(\lambda,\mu)$-smoothness property with respect to a sequence of mechanisms $\M^n$, then $U_i^{n,\delta}=\E_{z}\left[z_i\cdot U_i^n(s\cdot z;v_i)\right]$ satisfies the $(\lambda,\mu)$-smoothness property with respect to the sequence of mechanisms $\M^{n,\delta}$.
\end{lemma}
\begin{proof}
By smoothness of $U_i^n$ with respect to $\M^n$ we know that for any valuation vector $v$, there exists for each player $i$ a deviation $s_i^*(v)$ such that for any strategy profile $s$:
\begin{equation*}
\sum_{i=1}^{n} U_i^n(s_i^*(v),s_{-i};v_i) \geq \lambda \opt^n(v) - \mu \rev^n(s)
\end{equation*}
Observe that for any strategy profile $s$:
\begin{equation*}
 U_i^n(s_i^*(v\cdot z),s_{-i};v_i)\cdot z_i \geq U_i^n(s_i^*(v\cdot z),s_{-i};v_i\cdot z_i)
\end{equation*}
Observe that:
\begin{align*}
\E_{z,\tilde{z}}\left[U_i^n( s_i^*(v\cdot (z_i,\tilde{z}_{-i})),s_{-i}\cdot z_{-i};v_i)\cdot z_i\right] =~& \E_{z,\tilde{z}}\left[U_i^n(s_i^*(v\cdot \tilde{z}),s_{-i}\cdot z_{-i};v_i)\cdot \tilde{z}_i\right]\\
\geq~&  \E_{z,\tilde{z}}\left[U_i^n(s_i^*(v\cdot \tilde{z}),s_{-i}\cdot z_{-i};v_i\cdot \tilde{z}_i)\right]
\end{align*}
Thus, for any strategy profile $s$, valuation profile $v$ and arrival vector $z$:
\begin{align*}
\sum_{i\in [n]} \E_{z,\tilde{z}}\left[U_i^n( s_i^*(v\cdot (z_i,\tilde{z}_{-i})),s_{-i}\cdot z_{-i};v_i)\cdot z_i\right]\geq ~&
\E_{z,\tilde{z}}\left[\sum_i U_i^n(s_i^*(v\cdot \tilde{z}),s_{-i}\cdot z_{-i};v_i\cdot \tilde{z}_i)\right]\\
\geq~& \E_{z,\tilde{z}}\left[\lambda \opt^n(v\cdot \tilde{z})-\mu \rev^n(s\cdot z)\right]
\end{align*}

Observe that: $\E_{\tilde{z}}\left[U_i^n(s_i^*(v\cdot (z_i,\tilde{z}_{-i})),s_{-i}\cdot z_{-i}; v_i)\cdot z_i\right]$ corresponds to the utility of a player under the following deviation: random sample an arrival vector $\tilde{z}_{-i}$, deviate assuming the arrival vector $(z_i,\tilde{z}_{-i})$. This is a valid deviation for the noisy arrival mechanism $\M^{n,\delta}$ and hence the above inequality shows that $U_i^{n,\delta}$ satisfies the $(\lambda,\mu)$-smoothness property with respect to $\M^{n,\delta}$.

\nicomment{move above paragraph before inequalities.}

\end{proof}

\paragraph{Approximation.} Now we move on to showing that $U_i^{n,\delta}$ approximates $u_i^{n,\delta}$ as $n\rightarrow \infty$.

\begin{lemma}\label{lem:simultaneous-approximation}
 For any valuation $v_i$ and for any bid profile sequence $b^n$:
\begin{equation}
\lim_{n\rightarrow \infty} \left\|u_i^{n,\delta}(b^n; v_i) -U_i^{n,\delta}(b^n; v_i)\right\|=0
\end{equation}
\end{lemma}
\begin{proof}
We need to show that for any $\epsilon$, there exists $n(\epsilon)<\infty$ such that for any $n>n(\epsilon)$ and for any bid profile $b$ and for any valuation $v_i$:
\begin{align*}
\Delta = \left\|u_i^{n,\delta}(b;v_i)-U_i^{n,\delta}(b;v_i)\right\|<\epsilon
\end{align*}
By triangle inequality we can lower bound the left hand side by:
\begin{align*}
\Delta \leq~& \left\| \E\left[v_i(x_i(b\cdot z))-v_i(X_i(b\cdot z))\right]\right\| \\
& + \left\|\E\left[\sum_{j\in [m]} X_i^j(b^j\cdot z)\theta_{k_j^n+1}(b_{-i}^j\cdot z_{-i})-x_i^j(b^j\cdot z)\cdot \theta_{k_j^n+1}(b^j\cdot z)\right]\right\|
\end{align*}
The first part of the upper bound can be upper bounded by:
\begin{align*}
\left\| \E\left[v_i(x_i(b\cdot z))-v_i(X_i(b\cdot z))\right]\right\| \leq~& H\cdot \Pr\left[x_i(b\cdot z)\neq X_i(b\cdot z)\right]\\
\leq~& H\cdot \sum_{j\in [m]} \Pr\left[x_i^j(b^j\cdot z)\neq X_i^j(b^j\cdot z)\right]
\end{align*}
The second part can also be upper bounded by the summation of the following two quantities:
\begin{align*}
 \left\|\E\left[\sum_{j\in [m]} \left(X_i^j(b^j\cdot z)-x_i^j(b^j\cdot z)\right)\theta_{k_j^n+1}(b_{-i}^j\cdot z_{-i})\right]\right\|\\
\left\|\E\left[\sum_{j\in [m]} x_i^j(b^j\cdot z)\cdot \left(\theta_{k_j^n+1}(b_{-i}^j\cdot z_{-i})-\theta_{k_j^n+1}(b^j\cdot z)\right)\right]\right\|
\end{align*}
The first quantity is upper bounded by:
\begin{align*}
B\cdot r\cdot \sum_{j\in [m]} \Pr\left[x_i^j(b^j\cdot z)\neq X_i^j(b^j\cdot z)\right]
\end{align*}
since by assumption all marginal bids fall in a range $[0,B]$ and the difference in the two allocations of a player is at most $r$, by the $r$-demand assumption.

The second quantity is upper bounded by:
\begin{align*}
r\cdot\sum_{j\in [m]} \left\|\E\left[\theta_{k_j^n+1}(b_{-i}^j\cdot z_{-i})-\theta_{k_j^n+1}(b^j\cdot z)\right]\right\|
\end{align*}
since a player is allocated at most $r$ units of each good $j$.

Thus, by the above reasoning, it suffices to show that there exists a finite $n(\epsilon)$ such that for any $n>n(\epsilon)$ the following two properties hold for each uniform price auction $j\in [m]$ and for any bid profile $b^j$
\begin{align*}
\Pr\left[x_i^j(b^j\cdot z)\neq X_i^j(b^j\cdot z)\right]\leq  \frac{\epsilon}{2m(B \cdot r+ H)}\\
\left\|\E\left[\theta_{k_j^n+1}(b_{-i}^j\cdot z_{-i})-\theta_{k_j^n+1}(b^j\cdot z)\right]\right\| \leq \frac{\epsilon}{2 r m}
\end{align*}
Hence we break the proof in two lemmas:
\begin{lemma}\label{lem:allocation-convergence} For any uniform price auction $j\in [m]$ and for any $\epsilon>0$, there exists $n(\epsilon)<\infty$ such that for any $n>n(\epsilon)$ and for any bid profile $b^j$:
\begin{equation}
\Pr\left[x_i^j(b^j\cdot z)\neq X_i^j(b^j\cdot z)\right]<\epsilon
\end{equation}
\end{lemma}
\begin{proof}
We will show that the probability converges to $0$ conditional on any draw of the random tie-breaking priority order. Moroever, conditional on the tie-breaking rule it is without loss of generality to assume that all marginal bids in $b^j$ are distinct and that there are no ties. The reason is that conditional on the tie-breaking priority rule, we can add small quantities (much smaller than the smallest difference between any two different marginal bids), to the input bids of the players, so as to simulate the exact same allocation rule as would have been achieved by the original bid profile and with the priority rule drawn (e.g. if player $i$ was ordered first by the tie-breaker then add to all his marginal bids $n\cdot \delta$, if he was ordered second then add $(n-1)\cdot \delta$ etc., similarly if any of his own bids are identical then add even smaller $\delta'$'s to differentiate them).

So suffices to prove that the probability goes to zero assuming that there are no two identical bids in $b^j$. Observe that the two allocations are different only when any of the marginal bids of player $i$ is among the $k+1-r$ and the $k+1$ highest arriving marginal bids. If a marginal bid is not among the $k+1-r$ and the $k+1$ highest arriving marginal bids, then it is either not allocated by both allocation rules because it is below the $k+1$ highest arriving bid or is allocated by both rules because it is among the $k-r$ highest arriving bids and so adding player $i$'s bids will not push it out of the allocation.

\nicomment{ above two paragraphs too informal. }

Let $B(b^j\cdot z; x)$ denote the number of arriving marginal bids that are strictly above $x$. Thus we can upper bound the desired probability by the union bound as:
\begin{align*}
\Pr\left[x_i^j(b^j\cdot z)\neq X_i^j(b^j\cdot z)\right]\leq \sum_{t=1}^{r} \sum_{q=k+1-r}^{k+1}\Pr[B(b^j\cdot z; b_i^{j,t})=q]
\end{align*}
Now we can use the re-interpretation of a Lemma by Swinkels re-written in our terminology:
\begin{lemma}[Swinkels \cite{Swinkels2001}] For any $x\in [0,B]$ and for any $\epsilon$, there exists a $q(\epsilon)\leq \infty$ such that for any $q>q(\epsilon)$ and for any $b^j$:
\begin{equation}
\Pr[B(b^j\cdot z; x)=q]\leq \epsilon
\end{equation}
\end{lemma}
Thus for sufficiently large $n$, $k_j^n$ is sufficiently large that each probability in the double summation can be made smaller than any $\epsilon$. Since the summation is over a constant number of quantities, the double summation can also be made smaller than any $\epsilon$ for sufficiently large $n$. This completes the proof of the Lemma.
\end{proof}

\begin{lemma}\label{lem:payment-convergence}
For any uniform price auction $j\in [m]$ and for any $\epsilon>0$, there exists $n(\epsilon)<\infty$ such that for any $n>n(\epsilon)$ and for any bid profile $b^j$:
\begin{equation}
\left\|\E\left[\theta_{k_j^n+1}(b_{-i}^j\cdot z_{-i})-\theta_{k_j^n+1}(b^j\cdot z)\right]\right\| \leq \epsilon
\end{equation}
\end{lemma}
\begin{proof}
Observe that $\theta_{k_j^n+1}(b_{-i}^j\cdot z_{-i})\leq \theta_{k_j^n+1}(b^j\cdot z)$. Moreover, since player $i$ submits at most $r$ bids, the $k_j^n+1$ highest bid among all bids except player $i$'s is at least the $k_j^n+1+r$ highest bid among all bids including player $i$'s. Thus:
\begin{align*}
\left\|\E\left[\theta_{k_j^n+1}(b_{-i}^j\cdot z_{-i})-\theta_{k_j^n+1}(b\cdot z)\right]\right\| \leq \left\|\E\left[\theta_{k_j^n+1+r}(b^j\cdot z)-\theta_{k_j^n+1}(b^j\cdot z)\right]\right\|
\end{align*}

We will now use the following reinterpretation of a Lemma of Swinkels, which we state in our terminology:
\begin{lemma}[Swinkels \cite{Swinkels2001}] For any $x,x'\in [k_j^n,k_j^n+r+1]$ the difference of the cummulative density functions of the $x$-th and the $x'$-th highest arriving bid in a single uniform price auction converges to $0$  uniformly over $x,x'$ and $b^j$, as $k_j^n\rightarrow \infty$.
\end{lemma}
Since the CDFs of the random variables $\theta_{k_j^n+1+r}(b^j\cdot z_{-i})$ and $\theta_{k_j^n+1}(b^j\cdot z)$ converge and since the two quantities are bounded in $[0,B]$, their expectations also converge and therefore there exists $n(\epsilon)$ such that for any $n>n(\epsilon)$ and for any $b^j$:
\begin{align*}
\left\|\E\left[\theta_{k_j^n+1+r}(b^j)\cdot z_{-i})-\theta_{k_j^n+1}(b^j\cdot z)\right]\right\| \leq \epsilon
\end{align*}
which completes the proof of the lemma.
\end{proof}

Combining Lemmas \ref{lem:allocation-convergence} and \ref{lem:payment-convergence} establishes the assertion of Lemma \ref{lem:simultaneous-approximation}.

\end{proof}

\subsection{Rates of Convergence for $r=1$}

\begin{theorem}
If $r=1$ and $k_j^n \geq \frac{36\cdot \rho^2\cdot m^2 (B+H)^2}{\epsilon^2 \delta(1-\delta)}$, then:
\begin{equation}
\textstyle{\sum_{i=1}^{n} u_i^{n,\delta}(s_{i}^{*,n},s_{-i}^n;v_i^n) \geq  (1-\epsilon) \opt^n(v^n)- \rev^n(s^n)}
\end{equation}
and therefore the robust price of anarchy is at most $\frac{1}{1-\epsilon}$.

Equivalently if $k_j^n = \Omega(n)$, and for constant $\rho, m,B,H,\delta$, the welfare at every equilibrium (BNE, CCE etc) is at least $\left(1-o\left(\frac{1}{\sqrt{n}}\right)\right)$ of the expected optimal welfare.
\end{theorem}

\begin{lemma}\label{lem:convergence-rate}If $r=1$, then for any $x\in [0,B]$ and for any $\epsilon$, if $q>\frac{4}{\epsilon^2\cdot\delta\cdot(1-\delta)}$, then for any $b^j$:
\begin{equation}
\Pr[B(b^j\cdot z; x)=q]\leq \epsilon
\end{equation}
\end{lemma}
\begin{proof}
Consider any bid profile $b^j$, consisting of one bid per player. If less than $q$ players are bidding above $x$ in $b^j$, then $\Pr[B(b^j\cdot z; x)=q]=0$ and the theorem follows. Thus in the bid profile that maximizes the probability that we want to upper bound, there are $t\geq q$ players bidding above $x$. Then the probability of the event of interest is equal to the probability that exactly $q$ of these players remain after the random deletion. Observe that the number of players among these bidders that remain after the random deletion follows a Binomial distribution of $t$ trials, each with success probability $(1-\delta)$, denoted as $\B(t,1-\delta)$. 

By the Berry-Esseen theorem \cite{Berry1941,Esseen1942,Shevtsova2011} we know that the CDF of $\B(t,p)$ is approximated by the CDF of the normal distribution with mean $t\cdot p$ and variance $t\cdot p\cdot (1-p)$, with an additive error that is upper bounded by $err\leq \frac{p^2+(1-p)^2}{2\sqrt{np(1-p)}}$. Denote with $\Phi(\cdot)$ the CDF of the standard normal distribution. If $X$ is a random variable distributed according to $\B(t,p)$, then 
\begin{align*}
\Pr[X=k] = \Pr[X\leq k]-\Pr[X\leq k-1] \leq~& \Phi\left(\frac{k-t\cdot p}{\sqrt{t\cdot p\cdot (1-p)}}\right) - \Phi\left(\frac{k-1-t\cdot p}{\sqrt{t\cdot p\cdot (1-p)}}\right)+2\cdot err\\
=~& \frac{1}{\sqrt{2\pi}} \int_{\frac{k-1-t\cdot p}{\sqrt{t\cdot p\cdot (1-p)}}}^{\frac{k-t\cdot p}{\sqrt{t\cdot p\cdot (1-p)}}} e^{-\frac{z^2}{2}}dz+2\cdot err\\
\leq~& \frac{1}{\sqrt{2\pi}} \frac{1}{\sqrt{t\cdot p\cdot (1-p)}}+2\cdot err\\
\leq~&  \left(\frac{1}{\sqrt{2\pi}}+p^2+(1-p)^2\right) \frac{1}{\sqrt{t\cdot p\cdot (1-p)}}\\
\leq~& \frac{2}{\sqrt{t\cdot p\cdot (1-p)}}
\end{align*}

By the above we get that:
\begin{equation}
\Pr[B(b^j\cdot z; x)=q]\leq \frac{2}{\sqrt{t\cdot \delta\cdot (1-\delta)}}\leq \frac{2}{\sqrt{q\cdot \delta\cdot (1-\delta)}}
\end{equation}
For $q\geq \frac{4}{\epsilon^2\cdot \delta\cdot (1-\delta)}$ the latter probability is at most $\epsilon$ as desired.
\end{proof}

\begin{lemma}\label{lem:payment-convergence}
For $r=1$, for any uniform price auction $j\in [m]$ and for any $\epsilon>0$, if $k_j^n \geq \frac{16B^2}{\epsilon^2 \delta(1-\delta)}$, then for any bid profile $b^j$:
\begin{equation}
\left|\E\left[\theta_{k_j^n+1}(b_{-i}^j\cdot z_{-i})-\theta_{k_j^n+1}(b^j\cdot z)\right]\right| \leq \epsilon
\end{equation}
\end{lemma}
\begin{proof}
Observe that $\theta_{k_j^n+1}(b_{-i}^j\cdot z_{-i})\leq \theta_{k_j^n+1}(b^j\cdot z)$. Moreover, since player $i$ submits one bid, the $k_j^n+1$ highest bid among all bids except player $i$'s is at least the $k_j^n+2$ highest bid among all bids including player $i$'s. Thus:
\begin{align*}
\left|\E\left[\theta_{k_j^n+1}(b_{-i}^j\cdot z_{-i})-\theta_{k_j^n+1}(b\cdot z)\right]\right| \leq \left|\E\left[\theta_{k_j^n+2}(b^j\cdot z)-\theta_{k_j^n+1}(b^j\cdot z)\right]\right|
\end{align*}

Let $F_t(\cdot)$ denote the CDF of the $t$-th highest bid. Observe that if the number of arriving bids strictly above $x$ are less than $t$, i.e., $B(b^j\cdot z; x)<t$, then both $\theta_t$ and $\theta_{t+1}$ are at most $x$ and therefore the conditional CDFs of $\theta_t$ and $\theta_{t+1}$ evaluated at $x$ are both $1$. If $B(b^j\cdot z; x)>t+1)$ then both $\theta_{t}$ and $\theta_{t+1}$ are strictly above $x$ and therefore the conditional CDFs evaluated at $x$ are both $0$. Thus the conditional CDFs differ only when $B(b^j\cdot z;x)\in [t,t+1]$ and they differ by at most $1$. Hence:
\begin{equation}
|F_t(x) - F_{t+1}(x)|\leq \Pr[B(b^j\cdot z;x)\in [t, t+1]]
\end{equation}
By Lemma \ref{lem:convergence-rate}, if $t\geq \frac{16B^2}{\epsilon^2\cdot \delta\cdot (1-\delta)}$, then $\Pr[B(b^j\cdot z;x)=t] \leq \frac{\epsilon}{2B}$ and $\Pr[B(b^j\cdot z;x)=t+1]\leq \frac{\epsilon}{2B}$, so by the union bound $|F_t(x) - F_{t+1}(x)|\leq \frac{\epsilon}{B}$.

Last observe that: 
\begin{equation}
\E\left[\theta_{t}(b^j\cdot z)-\theta_{t+1}(b^j\cdot z)\right] = \int_{0}^{B}1-F_t(x) dx - \int_0^B 1-F_{t+1}(x) dx =  \int_{0}^{B} F_{t+1}(x) -F_{t}(x) dx\leq \epsilon
\end{equation}
\end{proof}

\begin{lemma}
If  $k_j^n\geq \frac{36\cdot m^2 (B+H)^2}{\epsilon^2 \delta(1-\delta)}$ then for any valuation $v_i$ and for any bid profile sequence $b^n$:
\begin{equation}
\left\|u_i^{n,\delta}(b^n; v_i) -U_i^{n,\delta}(b^n; v_i)\right\|\leq\epsilon
\end{equation}
\end{lemma}
\begin{proof}
By same reasoning as in Lemma \ref{lem:simultaneous-approximation}, the difference in utilities is upper bounded by the following quantity:
\begin{align*}
 \sum_{j\in [m]} (B+H)\Pr\left[x_i^j(b^j\cdot z)\neq X_i^j(b^j\cdot z)\right]+
\left|\E\left[\theta_{k_j^n+1}(b_{-i}^j\cdot z_{-i})-\theta_{k_j^n+1}(b^j\cdot z)\right]\right|
\end{align*}
and:
\begin{align*}
\Pr\left[x_i^j(b^j\cdot z)\neq X_i^j(b^j\cdot z)\right]\leq \Pr[B(b^j\cdot z; b_i^{j,t})=k_j^n]+ \Pr[B(b^j\cdot z; b_i^{j,t})=k_j^n+1].
\end{align*}
By the previous lemmas, if $k_j^n\geq \frac{36\cdot m^2 (B+H)^2}{\epsilon^2 \delta(1-\delta)}$ then:
\begin{align*}
\left|\E\left[\theta_{k_j^n+1}(b_{-i}^j\cdot z_{-i})-\theta_{k_j^n+1}(b^j\cdot z)\right]\right|\leq~& \frac{\epsilon}{3\cdot m}\\
\Pr\left[x_i^j(b^j\cdot z)\neq X_i^j(b^j\cdot z)\right]\leq~& \frac{2\epsilon}{3m(B+H)}
\end{align*}
which subsequently gives that the utility difference is at most $\epsilon$.
\end{proof}

\subsection{Convergence Rate for General $r$}

\begin{theorem}
If $k_j^n \geq  \frac{16\cdot m^2 (B+H)^2 r^8 \rho^2}{\epsilon^2 \delta(1-\delta)}+r$, then:
\begin{equation}
\textstyle{\sum_{i=1}^{n} u_i^{n,\delta}(s_{i}^{*,n},s_{-i}^n;v_i^n) \geq  (1-\epsilon) \opt^n(v^n)- \rev^n(s^n)}
\end{equation}
and therefore the robust price of anarchy is at most $\frac{1}{1-\epsilon}$.

Equivalently if $k_j^n = \Omega(n)$, and for constant $\rho, r, m,B,H,\delta$, the welfare at every equilibrium (BNE or CCE) is at least $\left(1-o\left(\frac{1}{\sqrt{n}}\right)\right)$ of the expected optimal welfare.
\end{theorem}

\begin{lemma}\label{lem:convergence-rate-general}For any $x\in [0,B]$ and for any $\epsilon$, if $q>\frac{4 r^2}{\epsilon^2\cdot\delta\cdot(1-\delta)}$, then for any $b^j$:
\begin{equation}
\Pr[B(b^j\cdot z; x)=q]\leq \epsilon
\end{equation}
\end{lemma}
\begin{proof}
Consider any bid profile $b^j$, consisting of at most $r$ bids per player. For $h\in[1,r]$ let $N_h$ denote the subset of players that under $b^j$, they submit $h$ bids above $x$ and denote with $n_h=|N_h|$. Observe that there must exist at least one $h^*\in [1,r]$ such that $n_{h^*}\geq \frac{q}{r^2}$. Otherwise, we get: $\sum_{h=1}^r n_h\cdot r <q$ and therefore, there are in total less than $q$ bids above $x$. Hence, the probability we want to upper bound is $0$.

Let $Z_{-N_{h^*}}$ denote the number of arriving bids from players outside of $N_h$ and $Z_{N_{h^*}}$ the number of arriving bids from players in $N_h$. By the independent arrival assumption, conditional on the bid profile $b^j$, these two random variables are independent. Thus:
\begin{equation}
\Pr[B(b^j\cdot z; x)=q]=\sum_{q'=1}^{q} \Pr[Z_{-N_{h^*}}=q']\cdot \Pr[Z_{N_{h^*}}=q-q']\leq \max_{z\in [1,q]} \Pr[Z_{N_{h^*}}=z]
\end{equation}
Let $X$ denote the number of players from $N_{h^*}$ that end up arriving. Observe that: $\Pr[Z_{N_{h^*}}=z]=\Pr[X=\frac{z}{h^*}]$, if $z$ is a multiple of $h^*$ and $0$ otherwise. Thus:
\begin{equation*}
\Pr[B(b^j\cdot z; x)=q]\leq \max_{x\in [1,\floor{q/h^*}]} \Pr[X=x]
\end{equation*}
If we denote with $\B(t,p)$ the binomial distribution of $n$ trials each with success probability $p$, then observe that $X\sim \B(n_{h^*},1-\delta)$.

By the Berry-Esseen theorem \cite{Berry1941,Esseen1942,Shevtsova2011} we know that the CDF of $\B(t,p)$ is approximated by the CDF of the normal distribution with mean $t\cdot p$ and variance $t\cdot p\cdot (1-p)$, with an additive error that is upper bounded by $err\leq \frac{p^2+(1-p)^2}{2\sqrt{np(1-p)}}$. Denote with $\Phi(\cdot)$ the CDF of the standard normal distribution. If $X$ is a random variable distributed according to $\B(t,p)$, then 
\begin{align*}
\Pr[X=k] = \Pr[X\leq k]-\Pr[X\leq k-1] \leq~& \Phi\left(\frac{k-t\cdot p}{\sqrt{t\cdot p\cdot (1-p)}}\right) - \Phi\left(\frac{k-1-t\cdot p}{\sqrt{t\cdot p\cdot (1-p)}}\right)+2\cdot err\\
=~& \frac{1}{\sqrt{2\pi}} \int_{\frac{k-1-t\cdot p}{\sqrt{t\cdot p\cdot (1-p)}}}^{\frac{k-t\cdot p}{\sqrt{t\cdot p\cdot (1-p)}}} e^{-\frac{z^2}{2}}dz+2\cdot err\\
\leq~& \frac{1}{\sqrt{2\pi}} \frac{1}{\sqrt{t\cdot p\cdot (1-p)}}+2\cdot err\\
\leq~&  \left(\frac{1}{\sqrt{2\pi}}+p^2+(1-p)^2\right) \frac{1}{\sqrt{t\cdot p\cdot (1-p)}}\\
\leq~& \frac{2}{\sqrt{t\cdot p\cdot (1-p)}}
\end{align*}

By the above we get that:
\begin{equation}
\Pr[B(b^j\cdot z; x)=q]\leq \frac{2}{\sqrt{n_{h^*}\cdot \delta\cdot (1-\delta)}}\leq \frac{2 r}{\sqrt{q\cdot \delta\cdot (1-\delta)}}
\end{equation}
For $q\geq \frac{4 r^2}{\epsilon^2\cdot \delta\cdot (1-\delta)}$ the latter probability is at most $\epsilon$ as desired.
\end{proof}

\begin{lemma}\label{lem:payment-convergence}
For any uniform price auction $j\in [m]$ and for any $\epsilon>0$, if $k_j^n \geq \frac{4B^2 r^4}{\epsilon^2 \delta(1-\delta)}$, then for any bid profile $b^j$:
\begin{equation}
\left|\E\left[\theta_{k_j^n+1}(b_{-i}^j\cdot z_{-i})-\theta_{k_j^n+1}(b^j\cdot z)\right]\right| \leq \epsilon
\end{equation}
\end{lemma}
\begin{proof}
Moreover, since player $i$ submits at most $r$ bids, the $k_j^n+1$ highest bid among all bids except player $i$'s is at least the $k_j^n+1+r$ highest bid among all bids including player $i$'s. Thus:
\begin{align*}
\left|\E\left[\theta_{k_j^n+1}(b_{-i}^j\cdot z_{-i})-\theta_{k_j^n+1}(b\cdot z)\right]\right| \leq \left|\E\left[\theta_{k_j^n+1+r}(b^j\cdot z)-\theta_{k_j^n+1}(b^j\cdot z)\right]\right|
\end{align*}

Let $F_t(\cdot)$ denote the CDF of the $t$-th highest bid. Observe that if the number of arriving bids strictly above $x$ are less than $t$, i.e., $B(b^j\cdot z; x)<t$, then the conditional CDFs of $t$ and $t+r$ highest bid evaluated at $x$ are both $1$. If $B(b^j\cdot z; x)>t+r$ then the conditional CDFs evaluated at $x$ are both $0$. Thus the conditional CDFs differ only when $B(b^j\cdot z;x)\in [t,t+r]$ and they differ by at most $1$. Hence:
\begin{equation}
|F_t(x) - F_{t+r}(x)|\leq \Pr[B(b^j\cdot z;x)\in [t, t+r]]
\end{equation}
By Lemma \ref{lem:convergence-rate}, if $t\geq \frac{4B^2 r^4}{\epsilon^2\cdot \delta\cdot (1-\delta)}$, then for all $x\in [t,t+r]$, $\Pr[B(b^j\cdot z;x)=x]\leq \frac{\epsilon}{Br}$, and by the union bound: $|F_t(x) - F_{t+r}(x)|\leq \frac{\epsilon}{B}$.

Last observe that: 
\begin{equation}
\E\left[\theta_{t}(b^j\cdot z)-\theta_{t+r}(b^j\cdot z)\right] = \int_{0}^{B}1-F_t(x) dx - \int_0^B 1-F_{t+r}(x) dx =  \int_{0}^{B} F_{t+r}(x) -F_{t}(x) dx\leq \epsilon
\end{equation}
\end{proof}

\begin{lemma}
If  $k_j^n\geq \frac{16\cdot m^2 (B+H)^2 r^8}{\epsilon^2 \delta(1-\delta)}+r$ then for any valuation $v_i$ and for any bid profile sequence $b^n$:
\begin{equation}
\left\|u_i^{n,\delta}(b^n; v_i) -U_i^{n,\delta}(b^n; v_i)\right\|=\epsilon
\end{equation}
\end{lemma}
\begin{proof}
By same reasoning as in Lemma \ref{lem:simultaneous-approximation}, the difference in utilities is upper bounded by the following quantity:
\begin{align*}
r\cdot \left(\sum_{j\in [m]} (B+H)\Pr\left[x_i^j(b^j\cdot z)\neq X_i^j(b^j\cdot z)\right]+
\left|\E\left[\theta_{k_j^n+1}(b_{-i}^j\cdot z_{-i})-\theta_{k_j^n+1}(b^j\cdot z)\right]\right|\right)
\end{align*}
and:
\begin{align*}
\Pr\left[x_i^j(b^j\cdot z)\neq X_i^j(b^j\cdot z)\right]\leq~& \sum_{t=1}^{r} \sum_{q=k_j^n+1-r}^{k_j^n+1}\Pr[B(b^j\cdot z; b_i^{j,t})=q]\\
 \leq~& r^2 \max_{x\in [0,B],q\in [k_j^n+1-r,k_j^n+1]} \Pr[B(b^j\cdot z; x)=q]
\end{align*}
By the previous lemmas, if $k_j^n\geq \frac{16\cdot m^2 (B+H)^2 r^8 }{\epsilon^2 \delta(1-\delta)}+r$ then:
\begin{align*}
\Pr\left[x_i^j(b^j\cdot z)\neq X_i^j(b^j\cdot z)\right]\leq~& r^2 \frac{\epsilon}{2mr^3(B+H)}\\
\left|\E\left[\theta_{k_j^n+1}(b_{-i}^j\cdot z_{-i})-\theta_{k_j^n+1}(b^j\cdot z)\right]\right|\leq~& \frac{\epsilon}{2\cdot r\cdot m}\\
\end{align*}
which subsequently gives that the utility difference is at most $\epsilon$.
\end{proof}

\subsection{Constant Inefficiency in the Limit under Supply Uncertainty}\label{sec:inefficiency-supply}
\begin{example}
\label{ex:inefficiency-supply}
Consider a simultaneous uniform price auction game with two types of goods $A$ and $B$. 
For each market size $n$, there are $t=n/3$ unit-demand players that have only a value of $1/2$ for each unit of good $A$ and no value for a unit of good $B$, i.e. $v_i(x_i) = \frac{1}{2} {\bf 1}\{x_i^1\geq 1\}$. We refer to these players as type $a$ players. 
There are also $t$ unit-demand players that have value $1/2$ only for units of good $B$ and not of good $A$ and we refer to them as type $b$ players.
Finally, there are $t$ unit-demand players each having a value of $1$ for each unit of each good and desiring only one unit from some of the two goods, i.e. $v_i(x_i) = {\bf 1}\{x_i^1+ x_i^2\geq 1\}$. We refer to them as type $c$ players.
The supply of each good is distributed uniformly in $[0,t]$. Obviously, as $t\rightarrow \infty$ this supply distribution satisfies the property that the supply being equal to any fixed number goes to $0$.

\paragraph{Equilibrium.} We argue that the following is an equilibrium: Each type $c$ player picks uniformly at random one good $A$ or $B$ and submits a bid of $1$ at the uniform price auction for that good. Each type $a$ player submits a bid of $1/2$ at the uniform price auction for good $A$ and each type $b$ player submits a bid of $1/2$ for good $B$.

\paragraph{Equilibrium verification.} This is obviously an equilibrium for the type $a$ and $b$ players, since they are essentially unit-demand players in a single uniform price auction, hence the mechanism is dominant strategy truthful from their perspective. Thus it remains to argue that type $c$ players don't want to bid on both items. At each item, for them to win a unit they have to bid at least $1/2$, or otherwise they will lose to the type $a$ or type $b$ players. Moreover, the uniform price that they will have to pay on each good is always at least $1/2$, since there are always $t$ players bidding $1/2$.

Moreover, observe that in the limit of many players we can essentially assume that exactly half of the type $c$ players go for item $A$ and exactly half go for item $B$ (in fact we could have also analyzed the less natural equilibrium where this happens deterministically for every $n$, assuming $n/3$ is even). A type $c$ player's utility at the current equilibrium strategy is $1/2$ when the supply of the good that he chose is less than or equal to $t/2$, since he pays $1/2$, and it is $0$  when the supply is more than $t/2$, since he pays $1$. Thus his expected utility is $1/4$.

When he bids on both items, then observe that whenever he wins a unit at both auctions he has to pay a price of $2\cdot 1/2$, thus getting $0$ utility. Thus he gets any utility only when he wins a unit at exactly one of the two auctions and only when he wins it at a price of $1/2$. There are only two possible reasonable bids for the player $1$ or $1/2$. When he bids $1/2$ he is tying with the type $a$ or $b$ players and ties are broken at random. Let $p(b)$ be the probability that a player wins at auction $A$ or $B$ with a bid of $b$. Observe that $p(1/2)< p(1)=1/2$. Thus if $b_A,b_B\in \{1/2,1\}$ are the bids of the player in each auction, then his utility is:
$p(b_A)(1-p(b_B))\frac{1}{2} + p(b_B)(1-p(b_A))\frac{1}{2}= \frac{1}{2}\left(p(b_A)+p(b_B)- p(b_A) p(b_B)\right)$. For $0\leq p(b)\leq 1/2$, the latter is maximized at $p(b_A)=p(b_B) = 1/2$, leading to a utility of $1/4$. In essence, the only reasonable bid of a type $c$ player is to submit $1$ on one of the two goods.

\paragraph{Sub-optimality.} Now we argue about the suboptimality of this equilibrium as $n\rightarrow \infty$. The optimal allocation is to give as many units of either good $A$ or $B$ as possible to type $c$ players and all remaining units to type $a$ or $b$ players. There are always enough remaining type $a$ or $b$ players for all units to be allocated. Thus the expected optimal welfare is:
\begin{align*}
\E[\opt(k_A,k_B)] =~& \E\left[\min\{k_A+k_B,t\}+\frac{1}{2} \left(k_A+k_B-t\right)^+\right]\\
=~&\E\left[k_A+k_B-\frac{1}{2} \left(k_A+k_B-t\right)^+\right] = t- \frac{1}{2}\E\left[\left(k_A+k_B-t\right)^+\right] \\
=~& t \left( 1- \frac{1}{2} \E\left[ \left(\frac{k_A}{t}+\frac{k_B}{t}-1\right)^+\right]\right)
\end{align*}
As $t\rightarrow \infty$, then $\frac{k_A}{t}$ and $\frac{k_B}{t}$ are distributed uniformly in $[0,1]$. Thus by simple integrations for two $U[0,1]$ random variables $x,y$: $\E[(x+y-1)^+]=1/6$. Hence, $\E[\opt(k_A,k_B)]\approx \frac{11\cdot t}{12}$.

On the other hand the expected welfare at equilibrium is simply:
\begin{align*}
\E[SW(b)] =~& 2\cdot \E\left[\min\{k_A,t/2\} + \frac{1}{2} \left(k_A-t/2\right)^+\right]\\
=~& 2\cdot \E\left[ k_A- \frac{1}{2} \left(k_A-t/2\right)^+\right] = t- \E\left[\left(k_A-t/2\right)^+\right]
\end{align*}
For large enough $t$, $\E\left[\left(k_A-t/2\right)^+\right]\approx \frac{t}{8}$. Therefore $\E[SW(b)]\approx \frac{7\cdot t}{8}$. Therefore, the ratio of the expected optimal welfare over the expected equilibrium welfare converge to $\frac{22}{21}>1$.  Hence the limit price of anarchy is strictly greater than $1$.

It is worth noting that this example can be taken to the extreme, when there are $m$ goods, $t$ type $c$ players are interested in all of the goods and each good has a set of $t$ price setters interested only in that good, with value $1/2$. The supply of each good is distributed uniformly in $\left[0,\frac{2t}{m}\right]$. One equilibrium is for the type $c$ players to pick one item uniformly at random and bid $1$, while the price setting people bid truthfully on their good. As $m$ grows large, then the total supply $\sum_{j} k_j$ is with high probability concentrated around it's expected value, which is $t$. Thus the expected optimal welfare converges to $t$. On the other hand at equilibrium each good has approximately $t/m$ type $c$ players and the supply of that good is distributed $U\left[0,\frac{2t}{m}\right]$. Thereby the expected welfare from each good from calculation similar to the two good case, is $\frac{7 t}{8 m}$. Hence, the price of anarchy converges to $8/7$. The essence is that bidders cannot take advantage of the concentration of total supply, which the optimal welfare can.
\end{example}

\section{Greedy Combinatorial Auctions}\label{sec:greedy}
\label{sec:greedy-app}
As in Section \ref{sec:simultaneous-uniform}, we consider a setting with $n$ bidders and and a fixed number of $m$ different (types of) goods.
For this section, we will focus on a restricted class of multi-unit single-minded valuation functions, which take the following form: each agent $i$ has a desired set of items $S_i \subseteq [m]$ and a non-convex function $v_i : \mathbb{N} \to [0,H]$, where $v_i(\ell)$ denotes agent $i$'s value for receiving $\ell$ copies of set $S_i$, up to a maximum of $r$.  
%
Write $d$ for the maximum size of any set $S_i$. 


The goods will be sold via a greedy auction.
Bidders submit bids, in the form of a desired set $T_i$ and a list of marginal values $b_i^1\geq \ldots\geq b_i^r$. The bids are then considered in decreasing order.\footnote{In the same way as in Simultaneous Uniform Price Auctions, we can handle ties in such a way that it is without loss to assume all bids are distinct.}  When a bid $b_i^\ell$ is considered, then one unit of each item in $T_i$ will be allocated to player $i$ if there are remaining units of \emph{all} items in $T_i$, 
otherwise the bid is rejected.  
We will write $x_i^{k}(b)$ for the number of copies allocated to player $i$ by this auction when the supply vector is $k = (k_1, \dotsc, k_m)$.

For payments, 
we will charge
each player $i$ an amount, per unit of set $T_i$ received, equal to the largest bid that a shadow player could have placed on set $T_i$ and been rejected.
We formalize this as follows:
%
%
choose an item $j$, fix the quantity $k_{-j}$ of all other items, and imagine that there are infinitely many copies of item $j$.  Denote with $\theta^t_j(b)$ the $t$-th highest bid for a set containing item $j$ that would be allocated on input $b$.  Write $\theta^{k}(T_i,b) = \max_{j \in T_i}\{\theta_j^{k_j+1}(b)\}$.  Then player $i$'s payment will be $x_i^{k}(b) \cdot \theta^{k}(T_i,b)$.
The utility of a player in the greedy auction, given supply profile $k$, is then
$u_i^{n,k}(b; v_i) = v_i\left(x_i^{k}(b)\right) - x_i^{k}(b) \cdot \theta^{k}(T_i, b).$
%

%
%


We now define a version of the auction in which there is an endogenously noisy supply of items.

\begin{defn}
We say that the sequence of markets satisfies \emph{supply uncertainty} if the quantity $k^n_j$ of item $j$ is a random variable, and moreover for any $\epsilon > 0$ there exists some $n(\epsilon)$ such that, for all $n > n(\epsilon)$, $\Pr[k_j(n) = t] < \epsilon$ for all $j$ and all values $t$.
\end{defn}

We will show that under this notion of supply uncertainty, the greedy combinatorial auction is approximately efficient in the limit.

\begin{theorem}[Approximate efficiency in the Limit]\label{thm:greedy}
The greedy combinatorial auction under supply uncertainty admits a $(1,d)$-smooth approximation in the large.  In particular, if $k_j(n) = \Omega(n)$ for each item $j$, then the implied sequence of mechanisms is $(1,d)$-smooth in the large, and hence achieves a $1/d$ fraction of the optimal welfare.
\end{theorem}

One might hope to prove an analogous result to Theorem~\ref{thm:greedy} under demand uncertainty as well as under supply uncertainty.  However, it turns out that under demand uncertainty, a bidder's proposed approximate utility and actual utility may fail to converge; see Appendix \ref{greedy.appendix.demand} for an example.  Thus, to apply our framework to prove smoothness in the large under demand uncertainty, one would need to find an alternate approximate utility sequence.

\textbf{Sketch of proof of Theorem \ref{thm:greedy}.}
We first define a notion of approximate utility, then establish that this approximation satisfies the properties of a $(1,d)$-smooth approximation in the large.
%
%
To define the approximate utility,
consider $\theta^{k}(T_i, b_{-i})$, which is the critical value for set $T_i$ if agent $i$ were not present.  Write $X_i^{k}(b) = \max\{ \ell : b_{i,\ell} > \theta^{k}(T_i, b_{-i})\}$.  That is, $X_i^{k}(b)$ is the number of bids made by agent $i$ that are strictly greater than $\theta^{k}(T_i, b_{-i})$.  Then the approximate utility is:
\begin{equation}\label{eqn:util-decomp-demand-greedy}
\textstyle{U_i^{n,k}(b; v_i) = v_i\left(X_i^{k}(b)\right) - X_i^{k}(b) \cdot \theta^{k}(T_i, b_{-i}).}
\end{equation}
This is the utility of the original game, not taking into account the effect of player $i$'s bid upon the critical value of $T_i$.
We denote by $u_i^n$ and $U_i^n$ the expected utility and approximate utility, respectively, in expectation over the distribution of $k$.

We must show that $U_i^{n}$ satisfies the conditions of being a $(1,d)$-smooth approximation to the critical greedy auction, in the large.
The fact that $U_i^n$ approximates $u_i^n$ follows from the supply uncertainty: the variation in critical price calculation is smoothed over by uncertainty in the number of units of each item.  The smoothness condition follows in a manner similar to the Simultaneous Uniform Price auctions: under $U_i^n$, each agent effectively views herself as a price-taker; the factor of $d$ is effectively due to the approximation factor of the greedy allocation algorithm.

\subsection{Proof of Theorem \ref{thm:greedy}}
\paragraph{$(1,d)$-Smoothness of Approximate Utility.} We will first show that the approximate utility $U_i^{n}$ satisfies the conditions of being a $(1,d)$-smooth approximation to the critical greedy auction, in the large.  We do this in two steps.  We first show that $U_i^{n}$ satisfies the smoothness condition with respect to the critical greedy auction, then show that it approximates the utility of the original game.


\begin{lemma}
\label{lem:greedy-smooth}
For each $n$, $U_i^{n}$ satisfies the $(1,d)$-smoothness property with respect to the greedy critical price auction.
\end{lemma}
\begin{proof}
Fix valuation profile $v$, and let $x^{*,k}$ denote the welfare-optimal allocation for supply $k$.  We will consider the utility of agent $i$ when declaring his true valuation $v_i$.  
We have
\begin{align*}
U_i^{n}(v_i,b_{-i}; v_i)
& = \E_{k}\left[ v_i\left(X_i^{k}(v_i, b_{-i})\right) - X_i^{k}(v_i, b_{-i}) \cdot \theta^{k}(T_i, b_{-i}) \right] \\
& = \E_{k}\left[ \sum_{\ell = 1}^{r} \left(v_i(\ell) - v_i(\ell - 1) - \theta^{k}(T_i, b_{-i})\right)^+ \right] \\
& \geq \E_{k}\left[ \sum_{\ell = 1}^{r} \left(v_i(\ell) - v_i(\ell - 1) - \theta^{k}(T_i, b)\right)^+ \right] \\
& \geq \E_{k}\left[ v_i\left(x_i^{*,k}\right) - x_i^{*,k} \cdot \theta^{k}(T_i, b) \right].
\end{align*}
Taking a sum over all $i$ and applying linearity of expectation, we have
\begin{align*}
\sum_i U_i^{n}(v_i,b_{-i}; v_i)
& \geq \opt^n(v) - \E_{k}\left[ \sum_i x_i^{*,k} \cdot \theta^{k}(T_i,b) \right].
\end{align*}
Since $\theta^{k(n)}(T_i,b) = \max_{j \in T_i} \theta_j^k(b) \leq \sum_{j \in T_i} \theta_j^k(b)$, we have
\begin{align*}
\sum_i U_i^{n}(v_i,b_{-i}; v_i)
& \geq \opt^n(v) - \E_{k}\left[ \sum_j \theta^{k}_j(b) \sum_{i : T_i \ni j} x_i^{*,k} \right] \\
& \geq \opt^n(v) - \E_{k}\left[ \sum_j \theta^{k}_j(b) \cdot k_j \right] \\
& \geq \opt^n(v) - d \cdot \E_{k}\left[ \sum_i x_i^{k}(b) \cdot \theta^{k}(T_i,b) \right] \\
& = \opt^n(v) - d \cdot \rev^n(b)
\end{align*}
as required, where in the last inequality we made use of the fact that $\theta^k(T_i,b) \geq \frac{1}{d}\sum_{j \in T_i}\theta_j^k(b)$, plus the fact that $\theta_j^{k(n)}(b) = 0$ if not all copies of item $j$ are allocated in $x^k(b)$.
\end{proof}


\paragraph{Approximation.} Now we show that $U_i^{n}$ approximates $u_i^{n}$ as $n$ grows large.

\begin{lemma}
\label{lem.approx.greedy}
For any valuation $v_i$ and for any bid profile sequence $b^n$:
\begin{equation}
\lim_{n\rightarrow \infty} \left\|u_i^{n}(b^n; v_i) -U_i^{n}(b^n; v_i)\right\|=0
\end{equation}
\end{lemma}
\begin{proof}
This proof closely follows the proof of Lemma \ref{lem:simultaneous-approximation}.  Our goal is to find an upper bound on $\left\|u_i^{n}(b^n; v_i) -U_i^{n}(b^n; v_i)\right\|$.
Applying the triangle inequality to the definition of $u_i^n$ and $U_i^n$, we have
\begin{align*}
\left\|u_i^{n}(b^n; v_i) -U_i^{n}(b^n; v_i)\right\|
& \leq \| \E_{k}[ v_i(x_i^{k}(b)) - v_i(X_i^{k}(b))] \| \\
& \quad + \| \E_{k}[x_i^{k}(b) \cdot \theta^{k}(T_i,b) - X_i^{k}(b) \cdot \theta^{k}(T_i,b_{-i}) ] \|
\end{align*}
We'll bound separately each of the two terms on the right hand side.  The first  can be bounded by
\begin{align*}
\| \E_{k}[ v_i(x_i^{k}(b)) - v_i(X_i^{k}(b))] \|
& \leq H \cdot \Pr[x_i^{k}(b) \neq X_i^{k}(b) ] 
\end{align*}
For the second term, we have
\begin{align*}
& \| \E_{k}[x_i^{k}(b) \cdot \theta^{k}(T_i,b) - X_i^{k}(b) \cdot \theta^{k}(T_i,b_{-i}) ] \| \\
\leq & \| \E_{k}[x_i^{k}(b) (\theta^{k}(T_i,b) - \theta^{k}(T_i,b_{-i}))] \| + \| \E_{k}[(X_i^{k}(b) - x_i^{k}(b)) \theta^{k}(T_i,b)] \| \\
\leq & r \cdot \| \E_{k}[\theta^{k}(T_i,b) - \theta^{k}(T_i,b_{-i})] \| + H \cdot r \cdot \Pr[x_i^{k}(b) \neq X_i^{k}(b) ] 
\end{align*}
Given these bounds, it suffices to show that, for all $\epsilon > 0$, there exists an $n(\epsilon)$ such that, for all $n > n(\epsilon)$, we have
\[ \Pr[x_i^{k}(b) \neq X_i^{k}(b) ] < \epsilon \]
and
\[ \| \E_{k}[\theta^{k}(T_i,b) - \theta^{k}(T_i,b_{-i})] \| < \epsilon. \]
We will complete the proof by establishing these bounds in separate lemmas.

\begin{lemma}
\label{lem.greedy.1}
For all $\epsilon > 0$ there exists $n(\epsilon)$ such that for all $n > n(\epsilon)$, and all $i$ and $b$, $\Pr[x_i^{k}(b) \neq X_i^{k}(b) ] < \epsilon$.
\end{lemma}
\begin{proof}
By the union bound, we have
\[ \Pr[x_i^{k}(b) \neq X_i^{k}(b) ] \leq \sum_j \sum_{\ell = 1}^r \Pr[ \theta_j^{k}(b) > b_{i,\ell} \geq \theta_j^{k}(b_{-i})]. \]
It therefore suffices to bound $\Pr[ \theta_j^{k}(b) > b_{i,\ell} \geq \theta_j^{k}(b_{-i})]$.
Fix the quantities of all items but $j$, and suppose there are infinitely many units of item $j$.  Among the marginal bids in $b$, consider the winning bids for sets containing $j$; let $(z_1 \geq z_2 \geq \dotsc)$ be those bids in decreasing order.  Note then that $\theta_j^{k}(b) = z_{k_j+1}$, and $\theta_j^{k}(b_{-i}) \geq z_{k_j+r+1}$ (as agent $i$ is allocated at most $r$ copies of item $j$).

Let $\ell$ be the unique index such that $z_\ell > b_{i,\ell} \geq z_{\ell+1}$.  We then have
\[ \Pr[ \theta_j^{k}(b) > b_{i,\ell} \geq \theta_j^{k}(b_{-i}) ] \leq \Pr[ \ell+1 \leq k_j \leq \ell+r ]. \]
The union bound combined with the definition of supply uncertainty implies that, for sufficiently large $n$, this probability is at most $\epsilon \cdot r$.  Taking an appropriate choice of $\epsilon$  completes the proof.
\end{proof}




\begin{lemma}
\label{lem.greedy.2}
For all $\epsilon > 0$ there exists $n(\epsilon)$ such that, for all $n > n(\epsilon)$, and for all $i, j$, and $b$,
\begin{equation}
\left| \E_{k}\left[\theta_j^{k}(b)\right] - \E_{k}\left[\theta_j^{k}(b_{-i})\right] \right| < \epsilon.
\end{equation}
\end{lemma}
\begin{proof}
Define values $(z_1 \geq z_2 \geq \dotsc)$ as in Lemma \ref{lem.greedy.1}.  Recalling that $\theta_j^{k}(b) = z_{k_j+1}$ and $\theta_j^{k}(b_{-i}) \geq z_{k_j+r+1}$, we have
\begin{align*}
\left| \E_{k}\left[\theta_j^{k}(b)\right] - \E_{k}\left[\theta_j^{k}(b_{-i})\right] \right|
& \leq \sum_{\ell \geq 1}(z_\ell - z_{\ell+r}) \cdot \Pr[k_j = \ell].
\end{align*}
Supply uncertainty implies that, for sufficiently large $n$, $\Pr[k_j = \ell] < \epsilon$ for all $\ell$ and hence
\begin{align*}
\left| \E_{k}\left[\theta_j^{k}(b)\right] - \E_{k}\left[\theta_j^{k}(b_{-i})\right] \right|
& < \sum_{\ell \geq 1}(z_\ell - z_{\ell+r}) \cdot \epsilon
 \leq \sum_{\ell = 1}^{r}z_\ell \cdot \epsilon
 < r H \epsilon.
\end{align*}
Taking an appropriate choice of $\epsilon$ therefore completes the proof.
\end{proof}


Applying Lemma \ref{lem.greedy.1} and Lemma \ref{lem.greedy.2} then completes the proof of Lemma \ref{lem.approx.greedy}.

\end{proof}


\subsection{Sensitivity of Prices to Individual Deviations under Noisy Arrival}
\label{greedy.appendix.demand}

One might hope to prove an analogous result to Theorem~\ref{thm:greedy} under demand uncertainty as well as supply uncertainty.  That is, if every bidder arrives with probability $(1-\delta)$, then is the greedy combinatorial auction approximately efficient in the large?  In this section we show a partial negative result along these lines.  In particular, we show via example that a bidder's proposed approximate utility and actual utility fail to converge.  Thus to apply our framework to prove smoothness-in-the-large for demand uncertainty, one would need to find an alternate approximate utility sequence.

The example consists of two items with $k$ units each.  There are $2^k$ ``large'' bidders who have a value of $1$ for receiving at least one unit of each item, one bidder who has a value of $b_0=2$ for receiving at least one unit of the second item, and two bidders with values of $b_1=1/2$ and $b_2=1/4$ for receiving at least one unit of the first item.  Consider the event that bidder $0$, bidder $2$, and at least $k$ large bidders show up.  This event has probability arbitrarily close to $(1-\delta)^2$, which is a constant.  Conditional on this event, bidder $1$, using the approximate utility for the greedy auction, imagines a critical price of $0$ for the first item and so computes that his approximate utility is $1/2$.  However, in reality, the price for the first item would be set by $b_2=1/4$, and so the actual utility that would be gained by bidder $2$ is $1/4$.  Removing the conditioning, we see there is a constant gap between the approximate utility and actual utility that persists even as the market grows large, violating the conditions of an approximate utility sequence.

\section{Smoothness in the Large for General Games}\label{app:general-games}
Here we present the smoothness in the large framework for the case of cost-minimization games (utility maximization games are a complete analogue). A cost minimization game $G^n$ consists of a set $N$ of $n$ players, a type space $T^n=T_1\times\ldots\times T_n$ specifying a set of potential types $T_i$ for each player, a strategy space $S^n=S_1\times\ldots\times S_n$ specifying a set of potential strategies $S_i$ for each player, and a cost function $c_i^n$ for each player.  In an instantiation of the game, each player $i$ has a type $t_i\in T_i$.  The type of a player determines his feasible strategy space $S_i(t_i) \subseteq S_i$.  Write $S(t) = S_1(t_1) \times \ldots\times S_n(t_n)$ for the set of admissible strategy profiles given type profile $t$.  The cost of a player depends on his type and the strategies of all players, $c_i^n:S^n\times T_i\rightarrow\Re$, denoted as $c_i^n(s;t_i)$. The objective is to select an outcome that minimizes the social cost.  
Let $\SC^n(s;t)$ denote the social cost at strategy profile $s$, i.e.,
\begin{equation}
\SC^n(s;t) = \sum_{i=1}^n c_i^n(s; t_i).
\end{equation}
The optimal cost for type profile $t$ will be denoted by $\opt^n(t)$; i.e., 
\begin{equation}
\opt^n(t) = \min_{s\in S^n}\SC^n(s;t)
\end{equation}

For a fixed game $G^n$, we can imagine a Bayesian setting in which player types are drawn independently from distributions; that is, for each player $i$ there is a distribution $\F_i$ over $T_i$, and we think of $t_i$ as being drawn independently from $\F_i$.  A strategy function for agent $i$ is a (possibly randomized) mapping from $T_i$ to $S_i$, which we think of as a specification of the action to use given a type.  A $\BCCE$ is a profile of strategy functions such that no single agent can decrease her expected cost (over randomization in types and strategies) by unilaterally modifying her strategy.  Formally, the profile of strategy functions $\mu$ is a $\BCCE$ if for all $i$, all types $t_i$, and all alternative strategies $s_i' \in S_i$, we have
\[ \E_{t_{-i} \sim \F_{-i}}[c_i^n(\mu_i(t_i), \mu_{-i}(t_{-i});t_i)] \leq \E_{t_{-i} \sim \F_{-i}}[c_i^n(s_i', \mu_{-i}(t_{-i});t_i)]. \]
Note that the non-Bayesian notion of Nash Equilibrium is a special case of the above, in which every distribution $\F_i$ is a point mass.

The Price of Anarchy (PoA) of game $G^n$ is the worst-case ratio between the expected optimal cost and the expected social cost at equilibrium, over all type distributions and all $\BCCE$. Formally, the Bayes-Nash Price of Anarchy of a game $G^n$ is 
\[ \BCCE\mbox{-}\poa^n = \max_{\F} \max_{s}\frac{\E_{t \sim \F}[\SC^n(s(t);t)]}{\E_{t \sim \F}[\opt^n(t)]} \]
where the maximum over strategies $s$ is taken over all $\BCCE$ for distribution profile $\F$. 

\paragraph{Sequence of games.}

We will typically work with a sequence of games $\Gn_{n=1}^{\infty}$, which will intuitively correspond to the original game growing large. When clear in the context, we will denote the sequence by $\Gn$ for the sake of brevity.

\paragraph{Smoothness in the large.}
For finite games, Roughgarden \cite{Roughgarden2009} introduced the notion of {\it smoothness} as a method for bounding inefficiency of equilibria.  The smoothness approach proceeds by exploring specific deviations, instead of characterizing the (potentially complex) structure of equilibria. We extend the notion of smoothness to large games.  Intuitively, a sequence of games is said to be \emph{smooth in the large} if there is a sequence of strategy deviations such that for any strategy profile, as the game grows, the total cost under the proposed deviations from the strategy profile minus the cost at that profile itself does not exceed much the optimal cost.

\begin{defn}[Smooth in the large] A sequence of cost-minimization games $\Gn$ is $(\lambda,\mu)$-smooth in the large if for any $\epsilon>0$, there exists $n(\epsilon)<\infty$, such that for any $n>n(\epsilon)$, for any $t^n\in T^n$, for each $i\in [n]$, there exists a strategy $s_i^{*,n}\in S_i(t_i^n)$, such that for any $w^n\in T^n$ and $s^n\in S^n(w^n)$:
\begin{equation}
\sum_{i=1}^{n} c_i^n(s_{i}^{*,n},s_{-i}^n;t_i^n) \leq \lambda (1+\epsilon) \opt^n(t^n)+\mu\cdot SC^n(s^n;w^n)
\end{equation}
\end{defn}

The following theorem shows that if a sequence of cost-minimization games is $(\lambda,\mu)$-smooth in the large, for some $\lambda \geq 1$ and $\mu < 1$, then its price of anarchy converges to $\frac{\lambda}{1-\mu}$. Moreover, it implies that for any sufficiently large but finite market of size $n$ the price of anarchy of all Bayes-Nash equilibria is at most a $1+\epsilon(n)$ multiplicative factor away from the limit price of anarchy, where the rate of convergence of $\epsilon(n)$ to $0$ is application specific and can be derived from the proof of smoothness in the large.

\begin{theorem} If a sequence of games is $(\lambda,\mu)$-smooth in the large then
\begin{equation*}
\lim\sup_{n\rightarrow \infty} \BCCE\text{-}\poa^n\leq \frac{\lambda}{1-\mu}.
\end{equation*}
I.e., for any $\epsilon$ there exists a market size $n(\epsilon)$ such that for any $n\geq n(\epsilon)$, every Bayes-Nash equilibrium of the game $G^n$ with type distributions $\F_1\times \ldots\times \F_n$ has expected social cost at most $(1+\epsilon)\frac{\lambda}{1-\mu}$ times the expected optimal cost.
\end{theorem}
\begin{proof}
By $(\lambda,\mu)$-smoothness in the large, we have that for any $\epsilon$ there exists a market size $n(\epsilon)$ such that for any $n\geq n(\epsilon)$
the game $G^n$ is a $\left(\lambda(1+\epsilon),\mu\right)$-smooth game (in the sense of \cite{Roughgarden2012,Syrgkanis2012}). Therefore, by the results in \cite{Roughgarden2012,Syrgkanis2012}, the $\BCCE$-$\poa^n$ is at most $(1+\epsilon)\frac{\lambda}{1-\mu}$. The assertion of the theorem then follows.
\end{proof}

\subsection{Main Technique: Smooth Approximate Cost Functions}
We present the notion of a $(\lambda,\mu)$-smooth approximate cost function sequence with respect to a sequence of normal form games $\Gn$.

\begin{defn}[Smooth approximate cost] Let $C_i^n: S^n\times T_i\rightarrow \R_+$ be a cost function for player $i \in [n]$, and let $C^n=(C_1^n,\ldots,C_n^n)$ be a vector of cost functions. A sequence $\Seq{C}$ is a \emph{sequence of $(\lambda,\mu)$-smooth approximate cost functions for the sequence of games $\Seq{G}$}  if the following two properties are satisfied:
\begin{enumerate}
\item {\bf(Approximation)}\label{prop:approximation} The approximate cost $C_i^n$ converges to the true cost $c_i^n$ uniformly over $s^n\in S^n$ and $t_i\in T_i$. I.e., for any $\epsilon$, there exists $n(\epsilon)<\infty$, such that for any $n>n(\epsilon)$, for any $i\in [n]$ and $t_i\in T_i$, and for any $s^n\in S^n$ with $s_i^n\in S_i(t_i)$:
\begin{equation}
\left|c_i^n(s^n; t_i) -C_i^n(s^n; t_i)\right|<\epsilon.
\end{equation}
\item {\bf(Smoothness)} For each game $G^n$ in the sequence, the approximate cost satisfies the following $(\lambda,\mu)$-smoothness property with respect to $G^n$:
For any $n$, for any $t\in T^n$, for any $i\in [n]$, there exists a strategy $s_i^{*,n}\in S_i$, such that for any type profile $w^n\in T^n$ and any strategy profile $s^n\in S^n(w^n)$:
\begin{equation}
\sum_{i=1}^{n} C_i^n(s_i^{*},s_{-i}^n;t_i) \leq \lambda \opt^n(t) +\mu\cdot \SC^n(s^n;w^n)
\end{equation}
\end{enumerate}
\end{defn}

We show that if if a sequence of games admits a $(\lambda,\mu)$-smooth approximate cost, and if its optimal social cost increases at the same asymptotic rate as the number of players, then that sequence of games is $(\lambda,\mu)$-smooth in the large.
\begin{theorem}
If a sequence $\Gn$ of cost-minimization games admits $(\lambda,\mu)$-smooth approximate cost functions, and $\opt^n(t)=\Omega(n)$, then the game sequence is $(\lambda,\mu)$-smooth in the large.
\end{theorem}
\begin{proof}
Since the game admits $(\lambda,\mu)$-smooth approximate cost functions $\Seq{C}$, we have that for any $n$ and $t^n\in T^n$ there exists strategies $s_i^{*,n}$ for each $i\in [n]$ such that, for any $w^n\in T^n$ and $s^n\in S^n(w^n)$,
\begin{equation*}
\sum_{i=1}^{n} C_i^n(s_i^{*,n},s_{-i}^n;t_i^n) \leq \lambda \opt^n(t^n) +\mu\cdot \SC^n(s^n;w^n).
\end{equation*}
By the approximation property of $C_i^n$ we have that for any $\epsilon$, there exists $n(\epsilon)<\infty$ such that for any $n>n(\epsilon)$: $c_i^n(s;t_i)\leq C_i^n(s;t_i)+\epsilon$ for any $t_i\in T_i$ and $s^n\in S^n$ with $s_i\in S_i(t_i)$. Thus:
\begin{equation*}
\sum_{i=1}^{n} c_i^n(s_i^{*,n},s_{-i}^n;t_i^n) \leq \lambda \opt^n(t^n) +\mu\cdot \SC^n(s^n;w^n) + n\cdot \epsilon.
\end{equation*}
Since $\opt^n(t)=\Omega(n)$, we can write $\opt^n(t^n)\geq \rho \cdot n$ for some $\rho>0$ and for sufficiently large $n$. Thus we get:
\begin{equation*}
\sum_{i=1}^{n} c_i^n(s_i^{*,n},s_{-i}^n;t_i^n) \leq \left(\lambda+\frac{\epsilon}{\rho}\right) \opt^n(t^n) +\mu\cdot \SC^n(s^n;w^n).
\end{equation*}
Therefore, for any $\delta>0$, we can pick $\epsilon$ appropriately small, such that $\lambda+\frac{\epsilon}{\rho}\leq \lambda(1+\delta)$, which would then yield the theorem.
\end{proof}

\section{Congestion Games Omitted Proofs}\label{sec:congestion-games}\label{app:congestion}



%
%
%
\begin{proofof}{Lemma \ref{lem:approx-cost}}
Assuming that the cost functions are continuous, then we immediately get that $C_i^n$ satisfies the \emph{approximation} 
requirement of a smooth approximation,
since
\[\lim_{n\rightarrow \infty} \|c_i^n(s;t_i) - C_i^n(s;t_i)\| \leq  \lim_{n\rightarrow \infty} \sum_{e\in s_i}\left\|c_e\left(\frac{n_e(s_{-i})}{n}+\frac{1}{n}\right)-c_e\left(\frac{n_e(s_{-i})}{n}\right)\right\| = 0.\]

Moreover, the following calculation shows that if the edge cost-functions satisfy the property that 
\begin{equation}
x\cdot c_e(y)\leq \lambda\cdot x\cdot c_e(x)+\mu\cdot y\cdot c_e(y)
\end{equation} 
for all $x$ and $y$, then $C_i^n$ also satisfies the \emph{smoothness} requirement of a smooth approximation.  To see this, take $s_i^{*,n}$ to be player $i$'s strategy in the cost-minimizing strategy profile for game $G^n$; we then have
\begin{align*}
\sum_i C_i^n(s_i^{*,n},s_{-i};t_i) =~& \sum_{i=1}^{n}\sum_{e\in s_i^{*,n}}  c_e\left(\frac{n_e(s_{-i})}{n}\right)
\leq  \sum_{i=1}^{n}\sum_{e\in s_i^{*,n}}  c_e\left(\frac{n_e(s)}{n}\right)\\
 =~& \sum_{e\in [m]} n_e(s^{*,n})  c_e\left(\frac{n_e(s)}{n}\right) \\
\leq~& \lambda \sum_{e\in [m]} n_e(s^{*,n})  c_e\left(\frac{n_e(s^{*,n})}{n}\right)  + \mu \sum_{e\in [m]}n_e(s) c_e\left(\frac{n_e(s)}{n}\right) \\
=~& \lambda \cdot \opt^n(t) + \mu \cdot \SC^n(s;w).
\end{align*}
We conclude that $C_i^n$ is a $(\lambda,\mu)$-smooth approximate cost function, as claimed.
\end{proofof}

%
%

\end{appendix}
\end{document}